\documentclass{article}
\usepackage[letterpaper, margin=1in]{geometry}
\usepackage{graphicx} % Required for inserting images
\usepackage{xcolor}
\usepackage{hyperref}
\hypersetup{colorlinks=true, linkcolor=blue, citecolor=green}
\usepackage{enumitem}
\usepackage{placeins}
\usepackage{setspace}
\usepackage{amsmath,amsfonts,amsthm,amssymb,xcolor,mathtools}
\usepackage{thm-restate}
\usepackage{comment, algorithm, algpseudocode}
\usepackage[style=alphabetic,natbib=true, maxbibnames=99]{biblatex}
\usepackage{tikz}
\usetikzlibrary{decorations.pathreplacing}
\usepackage{dsfont}

\newtheorem{theorem}{Theorem}[section]
\newtheorem{lemma}[theorem]{Lemma}
\newtheorem{corollary}[theorem]{Corollary}

\newtheorem{definition}[theorem]{Definition}
\newtheorem*{remark}{Remark}

\newcommand{\F}{\mathbb{F}}

\newcommand{\OPT}{\operatorname{OPT}}
\newcommand{\lb}{\left(}
\newcommand{\rb}{\right)}

\newcommand{\veps}{\varepsilon}

\newcommand{\bb}{\mathbb}

\newcommand{\E}{\mathbb E}

\newcommand{\lam}{\lambda}
\newcommand{\Lam}{\Lambda}
\newcommand{\Lk}{\Lambda_k}

\newcommand{\Dr}{\Delta_R}
\newcommand{\Dl}{\Delta_L}
\newcommand{\D}{\operatorname{D}}

\newcommand{\ph}{\varphi}

\newcommand{\bone}{\mathbf{1}}
\newcommand{\Skr}{S_k^{(r)}}

\newcommand{\Span}{\operatorname{span}}

\newcommand{\one}{\mathbf 1}
\newcommand{\calH}{\mathcal H}

\newcommand{\lf}{\lfloor}
\newcommand{\rf}{\rfloor}

\title{Concentration from Product Moments via an Additional Element of Randomness}

\author{
Michael Saks\thanks{Department of Mathematics, Rutgers University, Piscataway, NJ 08854.}
\and
Aravind Srinivasan\thanks{Department of Computer Science, University of Maryland, College Park, MD 20742.}
\and
Renata Valieva\thanks{Department of Mathematics, University of Maryland, College Park, MD 20742.}
}
\date{}

\begin{document}

\maketitle

\begin{abstract}
The standard method of exponential moments for proving concentration bounds can often be replaced by an argument based on elementary symmetric polynomials. We introduce an additional element of randomness into this framework, which reduces the problem to bounding product moments over a uniformly sampled set of indices.

We show that this approach gives useful bounds in three settings. For read-$\Delta$ families under limited independence, we obtain bounds governed by the degrees of randomly induced dependency subgraphs, improving the dependence on worst-case degrees. For random linear hashing over arbitrary finite fields with (semi-)random inputs, we derive fixed-bin and maximum-load bounds by controlling the rank defect of random tuples of input keys. Finally, for stochastic processes, we show how decay of product moments yields concentration bounds, recovering the spectral and mixing-time scales for finite-state Markov chains.
   % The standard method of exponential moments for concentration bounds can often be replaced by the use of the elementary symmetric polynomials; we add an element of randomness to the latter, that reduces concentration to bounding product moments along a uniformly-sampled set of indices.
   % We show that this approach gives useful bounds in three settings: read-$\Delta$ families under limited independence, linear hashing with (semi-)random inputs, and  stochastic processes.
      % We then show that along with various additional ideas that we develop, this template gives useful bounds in three settings: read-$\Delta$ families under limited independence, linear hashing with (semi-)random inputs, and  stochastic processes.
\end{abstract}
\tableofcontents

\bigskip

\section{Introduction}
% general idea; SSS rootes
The Chernoff--Hoeffding bounds~\cite{chernoffbound, hoeffdingbound} are among the basic tools in concentration inequalities and randomized algorithms. In their standard form, they are proved by applying Markov's inequality to the exponential moment $\E[e^{\lam X}]$ of a sum $X=X_1+\ldots+X_n$ of independent bounded random variables. Many algorithmic applications, however, produce random variables that are not fully independent: the variables may satisfy only limited independence, may be negatively dependent, may arise from a dependent rounding scheme, or may be functions of a smaller set of independent random variables. Extensions of Chernoff-type bounds have been developed in several such settings, including limited independence~\cite{siamdm/SchmidtSS95,DBLP:conf/focs/BellareR94,approx/Skorski22}, negative dependence~\cite{DBLP:journals/rsa/DubhashiR98, siamcomp/PanconesiS97}, read-\(k\) families~\cite{rsa/GavinskyLSS15}, and dependent rounding methods such as level-set rounding, bipartite dependent rounding, pipage rounding, swap rounding, and randomized/iterated rounding frameworks~\cite{Srinivasan2001DistributionsOL, jacm/GandhiKPS06, focs/ChekuriVZ10, soda/HarveyO14, DBLP:conf/stoc/Bansal19}.

A particularly relevant starting point for us is the work of Schmidt, Siegel, and Srinivasan~\cite{siamdm/SchmidtSS95}.  Instead of using the full exponential moment $\E\left[e^{\lam \sum_i X_i}\right]$ to bound the tail of $\sum_i X_i$, they use a low-degree polynomial with non-negative coefficients whose expectation can be bounded under limited independence. As one consequence, Chernoff--Hoeffding type bounds can still be obtained when the variables are only \(k\)-wise independent, for an appropriate choice of \(k\).

\subsection{Product-moments framework}
\label{sec:prod-mom-framework}
The starting point of our framework is the same low-degree viewpoint, but we phrase it directly in terms of product moments.  For simplicity, suppose first that $X_i\in\{0,1\}$---the same reduction applies for variables in
\([0,1]\) as well.  We state the framework for upper tails.  In the applications below,
the corresponding lower-tail statements can be obtained by applying the same
argument to the complementary variables $(1-X_i)$, so we do not state them
separately.

Let
$$
X=X_1+\cdots+X_n.
$$
If $X\ge a$, then at least $\binom{a}{k}$ products of the form $X_{i_1}\cdots X_{i_k}$ are equal to one.  Therefore, 
$$
\Pr[X\ge a]
\le
\frac{\E S_k(X_1,\ldots,X_n)}{\binom{a}{k}}
$$
by Markov's inequality, where
$$
S_k(X_1,\ldots,X_n)
=
\sum_{1\le i_1<\cdots<i_k\le n}
X_{i_1}\cdots X_{i_k}
$$
is the $k$-th elementary symmetric polynomial.  Equivalently, if $I$ is a uniformly random $k$-sized subset of $[n]$, that we sometimes refer to as a \emph{$k$-subset}, then
$$
\E S_k(X_1,\ldots,X_n)
=
\binom{n}{k}\,
\E_I\E\prod_{i\in I}X_i.
$$
Thus, the upper-tail problem reduces to bounding the product moment of the variables indexed by a random small set $I$. Note that under $k$-wise independence, every $k$-tuple behaves well, and the underlying random variables $(X_i)_{i\in I}$ are independent, implying that $\E[\prod_{i\in I}X_i] = \prod_{i\in I}\E[X_i]$---which recovers the usual Chernoff bounds for an appropriately chosen $k$~\cite{siamdm/SchmidtSS95}. In more general dependent settings, however, it is too much to ask that every $k$-tuple is independent.  The natural question instead is whether a \emph{typical} random $k$-tuple behaves well, which is why averaging over the sampled index set is a useful viewpoint. 

We show that this simple addition of an element of randomness, via the random choice of $I$, leads---in conjunction with additional tools that we develop---to a variety of applications. As an example, we show in Section~\ref{sec:RWs} that if we write $I = \{i_1, i_2, \ldots, i_k\}$ where $i_1 < i_2 < \cdots < i_k$, then the sequence of \emph{gaps} 
\[ (i_1, i_2 - i_1, i_3 - i_2, \ldots, i_k - i_{k-1}) \]
is \emph{negatively associated}, and use this to develop new tail bounds for a family of stochastic processes. Thus, the fact that $I$ can be viewed as \emph{random}, turns out to be nontrivially useful. 

\subsection{Read-$\Delta$ families under Limited Independence}
% *discussion of previous work/motivation* Initialized by Gav; a different technique in Duppala. bounds are essentially optimal; of interest if Xi's are not fully independent
In this setting, we study concentration for a sum of dependent random variables
\(Y_1,\ldots,Y_n\).  Each \(Y_i\), viewed as an ``output", is a function of some subset of underlying
random variables \(X_1,\ldots,X_m\), which we assume satisfy a limited-independence
condition.  We say that \(Y_1,\ldots,Y_n\) form a read-\(\Delta\) family if each
underlying variable \(X_j\) influences at most \(\Delta\) of the variables \(Y_i\).

This natural model was introduced by Gavinsky et al.~\cite{rsa/GavinskyLSS15}, who proved Chernoff-type bounds for read-\(\Delta\) families.  Their result shows that if \(Y_1,\ldots,Y_n\) are Boolean read-\(\Delta\)
variables over \textbf{independent} underlying random variables \(X_1,\ldots,X_m\),
then the usual Chernoff exponent degrades by a factor of \(\Delta\):
\[
        \Pr\left[\sum_i Y_i \ge \alpha n\right]
        \le
        \exp\left(-\frac{n}{\Delta}D(\alpha\|p)\right),
\]
where \(p\) is the average of the expectations of the \(Y_i\)'s and $\exp(x)$ denotes $e^x$; 
also, for $p,q\in(0,1)$, we have the usual KL-divergence definition 
$$
D(q\|p):=
q\ln \frac{q}{p}
+
(1-q)\ln \frac{1-q}{1-p}.
$$
This loss of a factor of $\Delta$ is tight in general: as an example take \(Y_1=\cdots=Y_\Delta=X_1\), \(Y_{\Delta+1}=\cdots=Y_{2\Delta}=X_2\), and so on.  Then the sum of the \(Y_i\)'s is exactly \(\Delta\) times a sum of only \(n/\Delta\) independent variables, so one cannot hope for an exponent better than \(n/\Delta\). This powerful bound is often a useful complement to other tail bounds for sums of dependent random variables such as \cite{Janson2004LargeDF}. 

The original proof of~\cite{rsa/GavinskyLSS15} is based on Shearer's lemma and an entropy method.  A later proof of the same read-\(\Delta\) inequality by Duppala et al.~\cite{duppala_et_al:LIPIcs.ITCS.2025.47} follows more directly from an exponential-moment method when the underlying random variables are Boolean: they show that if \(F_1,\ldots,F_n\) are nonnegative and read-\(\Delta\), then
\[
        \E\left[\prod_{i=1}^n F_i\right]
        \le
        \left(\prod_{i=1}^n \E[F_i^\Delta]\right)^{1/\Delta}.
\]
Applying this with \(F_i=\exp(\lambda Y_i)\) for a suitable $\lambda$ gives the same Chernoff exponent \(n/\Delta\).

\paragraph{Limited independence in the underlying variables.}
The next natural question is whether concentration still holds when the underlying variables \(X_1,\ldots,X_m\) are not fully independent.  This is the setting studied in this work.  We assume that the \(X_j\)'s are only \(r\)-wise independent, while the output variables still form a read-\(\Delta\) family over the \(X_j\)'s.

A direct argument gives a first baseline.  If each output \(Y_i\) depends on at most \(\Delta_R\) underlying $X_j$ and $X_j$ influences at most \(\Delta_L\) outputs, then a set of \(k\) output variables depends on at most \(k\Delta_R\) underlying variables.  Thus, as long as \(k\le r/\Delta_R\), the relevant underlying variables are fully independent, and the read-\(\Delta_L\) Chernoff argument applies.  
Theorem~\ref{thm: general baseline} shows that this simple baseline yields
\[
 \Pr\left[\sum_i Y_i\ge \alpha n\right]
        \le
        \exp\left(
            -\frac{\lfloor r/\Delta_R\rfloor}{\Delta_L}D(\alpha\|p)
        \right).
\]

However, this baseline is often pessimistic. Using our framework, we are only concerned with the result on average for a \emph{uniformly random} $k$-tuple of $Y$'s.  Specifically,
the relevant object is the random subgraph induced by the $k$ randomly-sampled right vertices in the bipartite dependency graph between the \(X\)'s and the \(Y\)'s.

This random induced subgraph can improve on the worst-case bounds in two ways. First, the sampled output variables may have much smaller read degree than $\Dl$. Even if some underlying variable influences \(\Delta_L\) outputs in total, it may influence only a few of the randomly sampled outputs. In that case, the read-\(\Delta_L\) loss in the Chernoff exponent can be replaced by the read degree observed in the sampled subgraph.

For example, even in the sparse $d$-regular case with \(m=n\), while the baseline argument gives an exponent of order $D(\alpha\|p)\cdot\frac{r}{d^2}$, our bound instead gives an exponent
$$
D(\alpha\|p)\cdot
\Omega\left(
\min\left\{
\frac{r}{d},
\sqrt{\frac{r}{d}\ln\frac{n}{r}}
\right\}
\right).
$$
In general, we show that this approach allows us to remove dependence on the worst-case read degree $\Dl$.

% {\color{magenta} Renata: we can also state this in the $(\Dl,\Dr)$-regular case and replace $r/d$ by $r/\Dr$ in the bound above to emphasize that the exponent no longer depends on $\Dl$.}

There is a second possible improvement. The $k$ sampled outputs may depend on many fewer underlying variables than the
crude upper bound \(k\Delta_R\). Since the product-moment argument only needs independence of the underlying variables in the induced subgraph, even when \(k>r/\Delta_R\), the
sampled outputs may still depend on at most \(r\) underlying variables. In this case, the product degree increases from \(r/\Delta_R\) to \(k\), and the corresponding contribution to the exponent improves by the same factor.  

Section~\ref{sec:read} gives a detailed description and additional regimes of interest. 

\subsection{Linear hashing under input randomness}
Hashing gives a natural setting where concentration bounds are needed. Consider the balls-and-bins problem in which \(m\) keys are mapped to \(n\) bins. The maximum load controls, for instance, the worst-case lookup time in hashing with chaining. Under a truly random hash function, the hash values of any fixed distinct keys are independent and uniform, and the usual balls-and-bins analysis gives the optimal maximum-load scale. Such a hash function is too expensive to store, so a long line of work studies simpler hash families, including universal hashing, \(k\)-wise independent hashing, tabulation-based hashing, and linear hashing~\cite{DBLP:conf/stoc/CarterW77, DBLP:journals/jcss/WegmanC81, DBLP:journals/siamcomp/Siegel04, DBLP:conf/stoc/PatrascuT11, DBLP:journals/jacm/AlonDMPT99}.

In this work, we focus on one of the simplest hash families, linear hashing over a finite field. Let $q$ be a prime power, let $n=q^\ell$, and choose a uniformly random linear map
$$
        h:\mathbb F_q^u\to \mathbb F_q^\ell .
$$
Equivalently, $h(v)=Av$ for a random $\ell\times u$ matrix $A$ over $\mathbb F_q$. This family is universal, and its values on linearly independent inputs are independent and uniform. It requires $O(u\log_2 n)$ bits to describe, although it is not $3$-wise independent in general. For $q>2$, distinct nonzero keys may be scalar multiples, so pairwise independence on all distinct nonzero inputs no longer holds.

The worst-case model fixes an arbitrary key set and then samples the hash function. For $q=2$, Jaber, Kumar, and Zuckerman~\cite{10.1145/3717823.3718208} showed that random linear hashing over $\F_2$ achieves the fully random maximum-load scale up to constants. In addition, they showed a tail bound of the form
$$
        \Pr[L_{\max}\ge R\OPT(m,n)]\le O(R^{-2}),
$$
where $\OPT(m,n)$ denotes the expected maximum-load. Bshouty~\cite{bshouty2026notesecondorderexpectedmaximumload} subsequently sharpened the maximum-load tail bound by optimizing the potential method of~\cite{10.1145/3717823.3718208}. He also gave an independent fixed-bin analysis: for a fixed bin \(y\), the load \(Load_h(y)=|B\cap h^{-1}(y)|\) satisfies a sharp tail bound of order \(2^{-a^2}\) at threshold \(2^a\). This fixed-bin estimate is essentially tight, but using it to control the maximum load requires a union bound over all $2^\ell$ bins, thus losing a factor of $2^\ell$. As a result, it is weaker than the potential-based maximum-load tail bounds, except at sufficiently large thresholds. For example, in the balanced case \(m=n=2^\ell\), the union bound becomes useful only once \(a\gtrsim \sqrt{\ell}\). We note that the fixed-bin counting argument, although originally stated only for $\F_2$, also extends to any finite field $\F_q$.
% Bshouty's work sharpens the maximum-load tail bound of~\cite{10.1145/3717823.3718208}; optimal bounds are not known here. % \textcolor{red}{Please check if my parenthesized sentence here is correct.} {\color{cyan} I believe so.}

For arbitrary $q$, Kumar and Mon~\cite{Kumar2026ListDL} used Furstenberg-set estimates to show that random linear hashing achieves the fully random expected maximum-load order in the balanced worst-case setting $m=n=q^\ell$, with constants depending on $q$.

\paragraph{Input randomness and semi-randomness.}
We ask what can be proved in an intermediate model where the hash family is still linear, but the input data themselves have some (weak) randomness. This is close in spirit to the block-source framework of Chor and Goldreich~\cite{CG88} and to the work of Chung, Mitzenmacher, and Vadhan~\cite{Chung2008WhySH}, where the data items may be correlated and far from uniform, but each new item retains enough conditional entropy given the past. Under suitable conditional entropy assumptions, simple universal hash functions can behave much like ideal hashing for several classical applications, including chaining, linear probing, balanced allocations, and Bloom filters~\cite{Chung2008WhySH}. Another recent example of entropy formulations in modeling semi-random inputs appears in \cite{kumar-etal-subset-entropy}.
% (See \cite{kumar-etal-subset-entropy} for another recent example of entropy formulations in modeling semi-random inputs.) 

A random linear map behaves like ideal hashing on linearly independent tuples: if \(v_1,\ldots,v_k\) are linearly independent, then \(h(v_1),\ldots,h(v_k)\) are independent uniform elements of \(\mathbb F_q^\ell\). Conversely, linear dependencies among the \(v_i\)’s impose the corresponding linear dependencies among their hash values. Thus the obstruction is rank defect: a sampled tuple of input keys behaves ideally only as long as it has span dimension close to $k$.

The cleanest sufficient condition is a pointwise conditional min-entropy assumption.  We say that the input sequence $R_1,\ldots,R_m\in\mathbb F_q^u\setminus\{0\}$ is a $p_{\rm src}$-pointwise block source if, for every $i$, every history $r_1,\ldots,r_{i-1}$ in the support of $R_{<i}$, and every possible key $x$,
\begin{equation}\label{def:pointwise-block-source}
        \Pr[R_i=x\mid R_1=r_1,\ldots,R_{i-1}=r_{i-1}]
        \le p_{\rm src}.
\end{equation}

For linear hashing, this pointwise condition is useful because it immediately gives a subspace-hit bound.  If $V\le \mathbb F_q^u$ has dimension $d$, then
$$
        \Pr[R_i\in V\mid R_{<i}]
        \le |V|p_{\rm src}
        \le q^d p_{\rm src}.
$$
In particular, after $j$ previously selected keys have been revealed, the probability that the next selected key lies in their span is at most $q^j p_{\rm src}$.

This fits naturally into the product-moment framework of Section~\ref{sec:prod-mom-framework}: to bound the upper tail of the load of a fixed bin $y$, set
$$
        X_i:=\mathbf 1_{\{h(R_i)=y\}},
        \qquad
        L_y=\sum_{i=1}^m X_i .
$$
Our upper tail for $L_y$ would be in terms of
$$
        \mathbb E_I \mathbb E \prod_{i\in I} X_i,
$$
where $I$ is a uniformly random $k$-subset of $[m]$.  Conditional on the realized input keys, this product moment is small whenever the selected keys have high rank: if the selected tuple has span dimension \(d\), then the probability that all selected keys land in \(y\) is at most \(n^{-d}\). Thus the upper-tail problem reduces to bounding the rank defect of a random small tuple of input keys, parametrized by \(p_{\rm src}\).

Our product-moment analysis applies over $\mathbb F_q$ for every prime power $q$, with constants uniform in $q$. To illustrate the resulting bounds and their comparison with prior work, we now assume $q=2$ and $n=2^\ell$ for the remainder of this subsection.

Thus, our contribution here sits between two benchmarks.  If the input keys are fully random, then a fixed bin has the usual binomial-type tail.  If the input set is worst-case, random linear hashing still has the correct maximum-load scale by~\cite{10.1145/3717823.3718208}, but the best-known tail bounds are weaker.  Our goal is to understand an intermediate regime: the hash family is still linear, but the input source has enough randomness that a random small tuple of keys is unlikely to have large rank defect.

% The comparison with previous work is summarized in Tables~\ref{tab:single} and~\ref{tab:maxload}. 
The fixed-bin result of~\cite{bshouty2026notesecondorderexpectedmaximumload} is the closest in spirit to ours: it counts ordered linearly independent tuples inside a heavy bin and applies Markov's inequality. Their estimate is sharper in the worst-case-input, fixed-bin setting, while our rank-defect analysis does better under semi-random-input assumptions. 
% In particular, given the bounds in Table~\ref{tab:single}, we see that our bound is comparable when $p_{\rm src} = n^{-1.5}$ and only gets better if we have more randomness.

% \begin{table}[t]
% \centering
% \small
% \begin{tabular}{c|c|c|c}
% \textbf{Reference} & \textbf{Hash family} & \textbf{Input model} & \textbf{Tail at threshold $2^a$} \\
% \hline
% Ideal hashing
% & truly random function
% & worst-case
% & $2^{-\Theta(a2^a)}$
% \\[1.2ex]

% \cite{bshouty2026notesecondorderexpectedmaximumload}
% & random linear map
% & worst-case
% & $ 2^{-a^2+O(1)}$
% \\[1.2ex]

% This work
% & random linear map
% & random source
% & $2^{\frac{-a^2}{2}-a\eta\ell+O(a)}$
% \end{tabular}
% \caption{Comparison of fixed-bin upper-tail bounds in the balanced case $m=n=2^\ell$ with $a\leq \ell$ and min-entropy parameter $\log\frac{1}{p_{\rm src}} = (1+\eta)\ell$.}
% \label{tab:single}
% \end{table}

To illustrate the comparison, consider the balanced case \(m=n=2^\ell\). Bshouty's fixed-bin estimate gives a bound of the form 
\begin{equation}
\label{eqn:semi-random-example-bshouty}
\Pr[L_y\geq 2^A]\leq 2^{-A^2+O(1)}.
\end{equation}
In our random-input model, the corresponding bound is of the form
\begin{equation}
\label{eqn:semi-random-example}
\Pr[L_y\geq 2^A]\leq 2^{\frac{-A^2}{2}-A\eta\ell+O(A)},
\end{equation}
where $\eta$ is defined through the min-entropy
$\log_2\frac{1}{p_{\rm src}}=(1+\eta)\ell$. For instance, when $p_{\rm src}=n^{-3/2}$,\footnote{Since the domain in hashing is often much larger than the range $[n]$, some mild randomness in the input seems reasonable to assume as in~\cite{Chung2008WhySH}, compared to the restrictive case where the universe has size $n$.} then \(\eta=1/2\), and the exponent in \eqref{eqn:semi-random-example} becomes
\[
-\frac{A^2}{2}-\frac{A\log_2 n}{2}+O(A).
\]
Hence, since $A \leq \log_2 n$, with $A \ll \log_2 n$ in most regimes of interest, our tail bound is at least as strong as, and often much stronger than, \(2^{-A^2}\), as the lower-order term ``$O(A)$" in (\ref{eqn:semi-random-example}) is not of much consequence.

\subsection{Stochastic processes}
Finally, we apply the product-moment framework to concentration for dependent stochastic processes, with Markov chains being one useful example.  Chernoff--Hoeffding bounds for Markov chains and expander walks have a long history~\cite{DBLP:conf/focs/Gillman93, DBLP:journals/cpc/Kahale97, Lezaud1998ChernofftypeBF, Len2004OptimalHB, DBLP:journals/cpc/Wagner08}
% Gillman, Kahale, Lezaud, Le'on--Perron, Healy, and Wagner
.  These bounds typically depend on the spectral gap: for a walk of length $t$ and observables of stationary mean $\mu$, one obtains tail bounds  of the form
$$
\Pr\left[\left|X-\mu t\right|\geq \delta\mu t\right]
\leq
C_\varphi
\exp\left(-\Omega\left((1-\lambda)\delta^2\mu t\right)\right)
$$
for $0\leq \delta\leq 1$, up to a factor depending on the initial distribution $\varphi$, with a corresponding linear dependence on $\delta$ for larger deviations, when $\delta>1$. This dependence on the spectral gap is asymptotically optimal, even for simple reversible chains---and our goal is not to improve it.

We recover concentration of the same scale from low-degree product moments.  For a random walk $(v_1,\ldots,v_t)$ and functions $f_i:V\to[0,1]$, the objective becomes to upper bound
\[
\E\left[\prod_{j=1}^k f_{i_j}(v_{i_j})\right]
\]
for any set of times $i_1<\cdots<i_k$.  In the spectral setting we prove such product bounds directly; the dependence on the gaps $i_{j+1}-i_j$ captures the decay of correlations/dependence along the walk. Since the product-moment reduction averages over a uniformly sampled set of times, these gaps are typically large enough to turn this decay into a useful bound on the averaged product moment, and hence into concentration.

We emphasize that this viewpoint does not have to be tied to Markov processes. What the product-moment argument relies on is a bound on the product moment above, with the bound improving as the indices $i_1<\cdots<i_k$ become more separated. For Markov chains, such estimates follow from spectral contraction, or from mixing properties. In other settings, analogous inputs may appear as weak coupling, decay of influence, or spatial mixing. For example, correlation decay for the hard-core model~\cite{DBLP:conf/stoc/Weitz06} and block-wise correlation decay for colorings of sparse random graphs~\cite{DBLP:conf/icalp/Yin14} provide graph-indexed analogues of this principle. 

We then return to the Markov-chain setting and show that our approach  recovers mixing-time-based bounds. 
Chung et al.~\cite{Chung2012ChernoffHoeffdingBF} proved a Chernoff--Hoeffding bound depending only on the total-variation mixing time $T=T(\epsilon)$, which is useful in settings where the spectral norm is unknown, hard to estimate, or uninformative.  In particular, for a walk started from an initial distribution $\varphi$ and observables of stationary mean $\mu$, their bound gives
\[
\Pr\left[X\geq (1+\delta)\mu t\right]
\leq
C_\varphi
\exp\left(-\frac{\delta^2\mu t}{72T}\right)
\]
for $0\leq\delta\leq1$, again with a linear dependence on $\delta$ for larger deviations.  They also show that the dependence on $T$ is optimal up to constants for mixing-time-only bounds.  Thus, without additional assumptions, one should not expect to improve the asymptotic dependence on $t/T$.

Following the approach in~\cite{Chung2012ChernoffHoeffdingBF} combined with the product-moment method,  we recover asymptotically the same bound
\[
\Pr\left[X\geq (1+\delta)\mu t\right]
\leq
C_\varphi
\exp\left(-\frac{\delta^2\mu t}{6(1-\mu)T}\right),\qquad 0<\delta\ll 1, 
\]
which improves upon the constant in the exponent of ~\cite{Chung2012ChernoffHoeffdingBF} if $\mu$ is not too close to $1$. 

\section{Preliminaries}
In this section we state some notation and auxiliary facts used throughout the paper.  The first
subsection gives the product-moment reductions that drive all three applications.
The remaining subsections record the standard tools needed for the read-\(\Delta\),
hashing, and Markov-chain settings, respectively.  Throughout the paper, \(\log\) denotes the base-\(2\) logarithm, and \(\ln\) denotes the natural logarithm.

\subsection{From product moments to upper tails}

Let $z_1,\ldots,z_n\in[0,1]$. For $0\leq k\leq n$, define the $k$-th
elementary symmetric polynomial by
$$
S_k(z_1,\ldots,z_n)
:=
\sum_{1\leq i_1<\cdots<i_k\leq n}
z_{i_1}\cdots z_{i_k}.
$$

\begin{lemma}[\cite{siamdm/SchmidtSS95}]\label{lem: z-cont}
Let $z_1,\ldots,z_n\in[0,1]$, and suppose
$
\sum_{i=1}^n z_i\geq a
$.
Then, for every integer $1\leq k\leq \lfloor a\rfloor$,
$$
S_k(z_1,\ldots,z_n)\geq \binom{\lf a\rf}{k}.
$$
Moreover, let $z:=\sum_{i=1}^nz_i$, then 
$$
S_k(z_1,\ldots,z_n)\geq\binom{\lf z\rf}{k} + (z-\lf z \rf)\cdot\binom{\lf z\rf}{k-1} =:\binom{z}{k}_{\mathrm{lin}}
,
$$
where $\binom{a}{b}:=0$ if $b<0$ or $b>a$ with $\binom{0}{0}=1$.
\end{lemma}

This immediately gives the basic product-moment reduction.

\begin{lemma}[\cite{siamdm/SchmidtSS95}]\label{lem: tail reduction}
Let $X_1,\ldots,X_n\in[0,1]$ be random variables, and set
$
X:=\sum_{i=1}^n X_i
$.
Then, for every $a>0$ and every integer $1\leq k\leq a$,
$$
\Pr[X\geq a]
\leq
\frac{\E S_k(X_1,\ldots,X_n)}{\binom{ a}{k}}.
$$
\end{lemma}
Equivalently, our \emph{basic template} leveraged throughout here is that if $I$ is a uniformly random $k$-subset of $[n]$, sampled
independently of $X_1,\ldots,X_n$, then
\begin{equation}\label{eq:sym-prod}
\Pr[X\geq a]
\leq
\frac{\E_I\E\left[\prod_{i\in I}X_i\right]}
{\binom{a}{k}/\binom{n}{k}}.
\end{equation}

We will also use the following exponential version of the same reduction.
\begin{lemma}\label{lem:sym-exp}
Let $X_1,\ldots,X_n$ be real-valued random variables, and set
$X:=\sum_{i=1}^n X_i$. Let $I$ be a uniformly random $k$-subset of $[n]$,
sampled independently of $X_1,\ldots,X_n$. Then, for every $\alpha\in\mathbb{R}$,
every integer $1\leq k\leq n$, and every $\lambda>0$,
\[
\Pr[X\geq \alpha n]
\leq
e^{-\lambda\alpha k}
\E_I\E\left[
    \exp\left(\lambda\sum_{i\in I}X_i\right)
\right].
\]
\end{lemma}

\begin{proof}
For a fixed realization of $X_1,\ldots,X_n$, by Jensen's inequality we have
\[
\E_I\left[
    \exp\left(\lambda\sum_{i\in I}X_i\right)
\right]
\geq
\exp\left(
    \lambda \E_I\left[\sum_{i\in I}X_i\right]
\right)
=
\exp\left(
    \lambda\frac{k}{n}\sum_{i=1}^n X_i
\right).
\]
Hence, on the event $X\geq \alpha n$,
\[
\E_I\left[
    \exp\left(\lambda\sum_{i\in I}X_i\right)
\right]
\geq e^{\lambda\alpha k}.
\]
Applying Markov's inequality completes the proof.
\end{proof}
\iffalse
\begin{lemma}\label{lem:sym-exp}
Let $X_1,\ldots,X_n\in\{0,1\}$, with
$
X:=\sum_{i=1}^n X_i
$. Let $1\leq k\leq \alpha n$ and $\lambda>0$. If $I$ is a
uniformly random $k$-subset of $[n]$, then
$$
% \label{eq:exp-sym}
\Pr[X\geq \alpha n]
\leq
e^{-\lambda\alpha k}
\E_I\E\left[
\exp\left(\lambda\sum_{i\in I}X_i\right)
\right].
$$
\end{lemma}
\begin{proof}
We have
$$
\E S_k(e^{\lambda X_1},\ldots,e^{\lambda X_n})
=
\binom{n}{k}
\E_I\E\left[
\exp\left(\lambda\sum_{i\in I}X_i\right)
\right].
$$
On the event $X\geq \alpha n$, at least $\lceil\alpha n\rceil$ of the variables $X_i$ are equal to $1$.
Hence, if $Z\sim\operatorname{Hypergeometric}(n,\lceil\alpha n\rceil,k)$, then
$$
S_k(e^{\lambda X_1},\ldots,e^{\lambda X_n})
\geq
\binom{n}{k}\E e^{\lambda Z}.
$$
By Markov's inequality,
$$
\Pr[X\geq \alpha n]
\leq
\frac{
\E_I\E\left[
\exp\left(\lambda\sum_{i\in I}X_i\right)
\right]
}{
\E e^{\lambda Z}
}.
$$
Since $\E Z\geq \alpha k$, Jensen's inequality gives
$$
\E e^{\lambda Z}\geq e^{\lambda\alpha k},
$$
which concludes the proof.
\end{proof}
\fi

\subsection{Read-$\Delta$ families}
\begin{definition}
Let $X_1,\ldots,X_m$ be random variables. For each $i\in[n]$, let $P_i\subseteq[m]$, and let $f_i:\{0,1\}^{P_i}\to[0,1]$ be a function of $X_{P_i}$. We say that $Y_i=f_i(X_{P_i})$, $i\in[n]$, form a read-$\Delta$ family if
$$
|\{i:j\in P_i\}|\le \Delta
$$
for each $j\in[m]$. In other words, each random variable $X_j$ influences at most $\Delta$ of $Y_i$'s.
\end{definition}

\begin{lemma}[{\cite{duppala_et_al:LIPIcs.ITCS.2025.47}}]\label{lem: read-Delta}
Let \(F_1,\ldots,F_n\) be a nonnegative read-\(\Delta\) family over independent Boolean random variables \(X_1,\ldots,X_m\). Then
$$
\mathbb{E}\left[\prod_{j=1}^{n}F_j\right]
\le
\left(\prod_{j=1}^{n}\mathbb{E}\left[F_j^\Delta\right]\right)^{1/\Delta}.
$$
\end{lemma}

\subsection{Pointwise block sources}
\begin{definition}[\cite{Chung2008WhySH}]\label{def:block-source}
Let $U$ be a finite universe. A random sequence
$R_1,\ldots,R_m\in U$ is a $p_{\rm src}$-pointwise block source if, for every
$i\in[m]$, every realization $r_1,\ldots,r_{i-1}$ in the support of $R_{<i}$,
and every $x\in U$,
$$
\Pr[R_i=x\mid R_1=r_1,\ldots,R_{i-1}=r_{i-1}]\leq p_{\rm src}.
$$
Equivalently, every block has conditional min-entropy at least
$\log(1/p_{\rm src})$.
\end{definition}

In the linear hashing application, $U=\F_q^u\setminus\{0\}$. If
$V\leq \F_q^u$ is a subspace of dimension $d$, then
\[
\Pr[R_i\in V\mid R_{<i}]\leq |V|p_{\rm src}\leq q^d p_{\rm src}.
\]
In particular, after $j$ previously selected vectors have been revealed, the
probability that the next vector lies in their span is at most $q^j p_{\rm src}$.

\subsection{Markov chains}
Let \(M\) be an ergodic Markov chain on a finite state space \([n]\), with
transition matrix \(M(x,y)\) and stationary distribution \(\pi\). We use the
function convention: for \(f:[n]\to\mathbb R\),
\[
(Mf)(x):=\sum_{y\in[n]}M(x,y)f(y).
\]
The stationary distribution satisfies
\(
\sum_{x\in[n]}\pi(x)M(x,y)=\pi(y)\),
\(y\in[n].
\)

\begin{definition}
For \(\varepsilon>0\), the \(\varepsilon\)-mixing time of \(M\) is
\[
T(\varepsilon)
:=
\min\left\{
t\ge 0:
\max_{x\in[n]}
\|M^t(x,\cdot)-\pi\|_{\operatorname{TV}}
\le \varepsilon
\right\},
\]
where the total variation distance for two probability distributions \(\nu,\eta\) on \([n]\), is defined as
\[
\|\nu-\eta\|_{\operatorname{TV}}
:=
\max_{A\subseteq[n]}|\nu(A)-\eta(A)|
=
\frac12\sum_{x\in[n]}|\nu(x)-\eta(x)|.
\]
\end{definition}
We write \(L_2(\pi)\) for the Hilbert space of functions
\(f:[n]\to\mathbb R\) with inner product and norm, respectively
\[
\langle f,g\rangle_\pi
:=
\sum_{x\in[n]}\pi(x)f(x)g(x),
\qquad
\|f\|_{2,\pi}
:=
\langle f,f\rangle_\pi^{1/2}.
\]
For a probability distribution \(\varphi\) on \([n]\), define
\(
\|\varphi\|_{\pi^{-1}}
:=
\left(\sum_{x\in[n]}\frac{\varphi(x)^2}{\pi(x)}\right)^{1/2}.
\)
Equivalently, if \(h=d\varphi/d\pi\), then
\(
\|\varphi\|_{\pi^{-1}}=\|h\|_{2,\pi}
\).
We also write
\(
\mathbb E_\pi f:=\sum_{x\in[n]}\pi(x)f(x)
\).

The constant function equal to one is denoted by \(\mathbf 1\). The subspace
orthogonal to \(\mathbf 1\) is precisely the set of mean-zero functions:
\[
\mathbf 1^\perp
=
\{f\in L_2(\pi):\mathbb E_\pi f=0\}.
\]
Since \(\pi\) is stationary, \(M\) preserves this subspace:
\(
\mathbb E_\pi Mf=\mathbb E_\pi f.
\)

\begin{definition}
Following \cite{Chung2012ChernoffHoeffdingBF}, define the spectral norm, or spectral expansion, of
\(M\) by
\[
\lambda(M)
:=
\sup_{\substack{f\in L_2(\pi),\ f\neq 0\\ \mathbb E_\pi f=0}}
\frac{\|Mf\|_{2,\pi}}{\|f\|_{2,\pi}}.
\]
Equivalently, \(\lambda(M)\) is the operator norm of \(M\) restricted to the
mean-zero subspace \(\mathbf 1^\perp\). Thus, for every mean-zero \(f\),
\[
\|Mf\|_{2,\pi}\le \lambda(M)\|f\|_{2,\pi}.
\]
\end{definition}
Iterating the last inequality gives
\[
\|M^g f\|_{2,\pi}\le \lambda(M)^g\|f\|_{2,\pi},
\qquad g\ge 0,
\]
for every mean-zero \(f\).

The \textit{time reversal} of \(M\) with respect to \(\pi\) is the transition matrix
\(\widetilde M\) defined by
\[
\widetilde M(x,y)
:=
\frac{\pi(y)M(y,x)}{\pi(x)}.
\]
Then \(\widetilde M\) is also a Markov chain with stationary distribution
\(\pi\). Moreover, \(\widetilde M\) is the \(L_2(\pi)\)-adjoint of \(M\):
\[
\langle f,Mg\rangle_\pi
=
\langle \widetilde M f,g\rangle_\pi,
\qquad f,g\in L_2(\pi).
\]
In particular, reversibility is exactly the condition
\(
M=\widetilde M,
\)
or equivalently,
\[
\pi(x)M(x,y)=\pi(y)M(y,x),
\qquad x,y\in[n].
\]

In other words, when \(M\)  is reversible, it is self-adjoint on \(L_2(\pi)\), and
\(\lambda(M)\) is the largest absolute value of a nontrivial eigenvalue of \(M\). Otherwise,
\(\lambda(M)^2\) is the largest nontrivial eigenvalue of \(\widetilde M M\) on the mean-zero subspace.

\begin{lemma}[\cite{Chung2012ChernoffHoeffdingBF}]\label{lem: M^T}
If \(T=T(\varepsilon)\), then
\(
\lambda(M^T)\le \sqrt{2\varepsilon}.
\)
In particular, for \(\varepsilon\le 1/8\),
\[
\lambda(M^T)\le \frac12.
\]
\end{lemma}

\subsection{Negative association}
We recall the notion of negative association and one standard fact that will be used in two places: in the read-\(\Delta\) analysis, and to average product bounds over randomly sampled time indices in the Markov-chain application.
\begin{definition}
A collection of random variables $X_1, \ldots, X_m$ is said to be negatively associated if for any $I, J \subset[m], I \cap J=\emptyset$ and any pair of non-decreasing (or non-increasing) functions $f:\bb R^I \rightarrow \mathbb{R}, g:\bb R^J \rightarrow \mathbb{R}$, $$ \mathbb{E}\left[f\left((X_i)_{i\in I}\right) g\left((X_j)_{j\in J}\right)\right] \leq \mathbb{E}\left[f\left((X_i)_{i\in I}\right)\right] \mathbb{E}\left[g\left((X_j)_{j\in J}\right)\right] . $$ 
\end{definition}

\begin{lemma}[\cite{joagdev1983negative}]\label{lemma: NA}
Let \(Y_0,\ldots,Y_k\) be independent integer-valued random variables with log-concave probability mass functions. Then the joint conditional distribution of \((Y_0,\ldots,Y_k)\) given \(\sum_jY_j=m\) is negatively associated.
\end{lemma}

\section{Read-$\Delta$ Families under Limited Independence}\label{sec:read}
This section proves the read-\(\Delta\) concentration bounds under limited independence. We begin with the regular case. Theorem~\ref{thm:d1-d2-simple} gives the direct baseline obtained by applying the read-\(d_L\) exponential-moment bound to a random \(k\)-subset, with \(k\) limited by the amount of independence. Theorem~\ref{thm:d1-d2-improved} then improves this when the randomly sampled subgraph typically has smaller maximum left-degree than the worst-case parameter \(d_L\), and Corollary~\ref{cor:regular improved dL} gives a concrete sparse-regime consequence. We then extend the argument to general dependency graphs in Theorem~\ref{thm:general-improvement}, where the two bad events are that the sampled outputs expose too many underlying variables or have too large an induced read degree. Corollaries~\ref{cor:general improved dR}, \ref{cor:general improved dL}, and~\ref{cor:general improved both} give explicit parameter choices illustrating the resulting improvements over the worst-case \(\Delta_L,\Delta_R\) bound. Corollary~\ref{cor:sunflower-comparison} compares the two neighborhood estimates on an explicit family of dependency graphs.

\subsection{The regular case}\label{sec:regular case}

\begin{theorem}\label{thm:d1-d2-simple}
Let $X_1,\ldots,X_m\in\{0,1\}$ be $r$-wise independent random variables, and let $Y_1,\ldots,Y_n\in[0,1]$ be a read-$d_L$ family with $\E[Y_i]=p$ for all $i\in[n]$. We discard underlying variables that appear in no output and assume $1\le r\le m$. Suppose that the dependency graph between the $X$'s and the $Y$'s is a $d_L$-$d_R$ regular bipartite graph, so that $d_Lm=d_Rn$. Assume $r\ge d_R$. Then, for every $0<p<\alpha<1$,
$$
\Pr\left[\sum_{i=1}^n Y_i\ge \alpha n\right]\leq \exp\left(-\frac{\left\lf r/d_R\right\rf}{d_L}\D(\alpha\|p)\right).
$$

\end{theorem}

\begin{proof}
Let $Y:=\sum_{i=1}^nY_i$. Following the product-moment method, let $I=\{i_1,\ldots,i_k\}$ be a uniformly random $\lf r/d_R\rf$-subset of $[n]$. The right-hand side vertices in $I$ are adjacent to at most $r$ left-hand side vertices. Hence, the corresponding left variables are independent by the $r$-wise independence assumption. Since the induced family $Y_{i_1},\ldots,Y_{i_k}$ is read-$d_L$, by Lemma~\ref{lem: read-Delta} we have
$$
\E\left[\exp\left(\lambda\sum_{j=1}^kY_{i_j}\right)\right]
\leq \lb\prod_{j=1}^k\E[e^{\lambda d_L Y_{i_j}}]\rb^{1/d_L}
\leq \left(1-p+pe^{\lambda d_L}\right)^{k/d_L},
$$
where $k = \lf r/d_R\rf$. Since $k\leq m/d_R = n/d_L\leq n$ we may apply Lemma~\ref{lem:sym-exp} to obtain
$$
\Pr[Y\ge \alpha n]\leq \frac{(1-p+pe^{\lambda d_L})^{k/d_L}}{e^{\lambda\alpha k}}.
$$
Optimizing over $\lambda>0$ in the standard way, the minimum is attained at
$$
\lambda=\frac{1}{d_L}\ln\frac{\alpha(1-p)}{p(1-\alpha)},
$$
and the resulting bound is
$$
\Pr[Y\ge \alpha n]\leq \exp\left(-\frac{k}{d_L}\D(\alpha\|p)\right).
$$
\end{proof}

\begin{restatable}{theorem}{regularReadImproved}\label{thm:d1-d2-improved}
Under the assumptions of Theorem~\ref{thm:d1-d2-simple}, fix $0<p<\alpha<1$. 
Assume that for some $d_0\in[d_L]$ the following inequality is satisfied
\begin{equation}\label{eq: exact log 1/eps}
d_0\ln\frac{m}{r}-\ln r+\ln((d_0+1)!)\geq \frac{r}{d_R}\left(\frac{1}{d_0}-\frac{1}{d_L}\right)\ln\frac{\alpha}{p}
\end{equation}
Then
$$
\Pr\left[\sum_{i=1}^nY_i\ge \alpha n\right]\leq 2\exp\left(-\frac{\lf r/d_R\rf}{d_0}\D(\alpha\|p)\right).
$$
In particular, whenever $d_0<d_L$, this improves the leading exponent of Theorem~\ref{thm:d1-d2-simple} by a factor of $d_L/d_0$.
\end{restatable}
\begin{proof}[Proof Sketch]
Let $k:=\left\lfloor r/{d_R}\right\rfloor$ and let $I$ be a uniformly random $k$-subset of right vertices, and let $\Lambda_k$ denote the maximum left-degree
in the subgraph induced by $I$. The variables indexed by $I$ depend on at most $r$
underlying variables, and hence the relevant variables are fully independent
by $r$-wise independence. On the event $\{\Lambda_k\leq d_0\}$, the induced
family is read-$d_0$; on the complementary event, we only use the trivial
read-$d_L$ bound.

Thus the exponential product-moment reduction, Lemma~\ref{lem:sym-exp}, gives
\[
\Pr\left[\sum_i Y_i\geq \alpha n\right]
\leq
(1-\veps_\Lambda)e^{-kF_{d_0}(\lambda)}
+
\veps_\Lambda e^{-kF_{d_L}(\lambda)},
\]
where
$
F_d(\lambda):=\lambda\alpha-\frac{1}{d}\ln(1-p+pe^{\lambda d})
$
and $\veps_\Lambda\geq \Pr[\Lambda_k>d_0]$. We evaluate this at the
Chernoff-optimal value for read degree $d_0$,
\[
\lambda_0=\frac1{d_0}\ln\frac{\alpha(1-p)}{p(1-\alpha)}.
\]
The first term is then at most
\[
\exp\left(-\frac{k}{d_0}\D(\alpha\|p)\right).
\]
It remains to ensure that the second term is no larger. This follows from
\[
\ln\frac1{\veps_\Lambda}
\geq
k\left(\frac1{d_0}-\frac1{d_L}\right)\ln\frac{\alpha}{p}.
\]
Finally, a union bound over left vertices gives
\begin{equation}\label{eq:exact thm3.2}
\ln\frac1{\veps_\Lambda}
\geq
d_0\ln\frac{m}{d_Rk}-\ln (d_Rk)+\ln((d_0+1)!),
\end{equation}
which is in turn at least the left-hand side of the assumed condition. The details of these two estimates
are deferred to Appendix~\ref{app:read-d-thm-regular}.
\end{proof}

\begin{remark}
If the assumption in Theorem~\ref{thm:d1-d2-improved} is strengthened to
$$
d_0\ln\frac{m}{r}-\ln r+\ln((d_0+1)!)\geq \frac{r}{d_R}\left(\frac{1}{d_0}-\frac{1}{d_L}\right)\ln\frac{\alpha}{p}+\omega(1),
$$
then the second term in the proof is $e^{-\omega(1)}$ times the first term. Hence the result improves to
$$
\Pr\left[\sum_{i=1}^nY_i\ge \alpha n\right]\leq (1+o(1))\exp\left(-\frac{\lf r/d_R\rf}{d_0}\D(\alpha\|p)\right).
$$
\end{remark}

The condition in Theorem~\ref{thm:d1-d2-improved} becomes intuitive if we ignore the lower-order term $\ln r$ and the harmless factor $1/d_0-1/d_L\le 1/d_0$. Then the requirement for $k:= \lf r/d_R\rf$ is roughly
$$
d_0\ln\frac mr
\gtrsim
\frac{k}{d_0}\ln\frac{\alpha}{p},
$$
or equivalently
$$
d_0^2
\gtrsim
\frac{k\ln(\alpha/p)}{\ln(m/r)}.
$$

The corollary below uses a concrete choice of $d_0$ for which the condition of Theorem~\ref{thm:d1-d2-improved} is satisfied.

\begin{corollary}\label{cor:regular improved dL}
Assume the setup of Theorem~\ref{thm:d1-d2-simple}, and suppose
\[
d_R=O(1),\qquad \ln(\alpha/p)=\Theta(1),\qquad
r\to\infty,\qquad r=o(m).
\]
Then, regardless of $d_L$,
\[
\Pr\left[\sum_{i=1}^nY_i\geq\alpha n\right]
\leq
2\exp\left(
-\D(\alpha\|p)\cdot
\Omega\left(\min\left\{r,\sqrt{r\ln(m/r)}\right\}\right)
\right).
\]
\end{corollary}

\begin{proof}
Set
\[
d_0:=\left\lceil
2\sqrt{\frac{(r/d_R)\ln(\alpha/p)}{\ln(m/r)}}
\right\rceil.
\]
Since $d_R=O(1)$ and $\ln(\alpha/p)=\Theta(1)$, we have
\[
\frac{\lf r/d_R\rf}{d_0}
=\Omega\left(\min\left\{r,\sqrt{r\ln(m/r)}\right\}\right).
\]
Thus, we just need to verify that the assumptions of Theorem~\ref{thm:d1-d2-simple} hold.

If $d_L\leq d_0$, Theorem~\ref{thm:d1-d2-simple} gives the claim.

Suppose therefore that $d_0<d_L$. By definition,
\[
d_0^2\ln\frac mr\geq4\frac r{d_R}\ln\frac{\alpha}{p}.
\]
Moreover, the assumptions imply
\[
\frac{(r/d_R)\ln(\alpha/p)\ln(m/r)}{\ln^2r}\longrightarrow\infty.
\]
Thus, for sufficiently large parameters,
\[
d_0\ln\frac mr
\geq2\sqrt{\frac r{d_R}\ln\frac{\alpha}{p}\ln\frac mr}
\geq2\ln r.
\]
Consequently,
\[
d_0\ln\frac mr-\ln r
\geq\frac12d_0\ln\frac mr
\geq\frac{r}{d_Rd_0}\ln\frac{\alpha}{p}.
\]
Since $\ln((d_0+1)!)\geq0$, the hypothesis of
Theorem~\ref{thm:d1-d2-improved} holds. Hence
\[
\Pr\left[\sum_{i=1}^nY_i\geq\alpha n\right]
\leq2\exp\left(-\frac{\lf r/d_R\rf}{d_0}\D(\alpha\|p)\right),
\]
which gives the stated bound.
\end{proof}

\subsection{The general degree case}\label{sec:general case}
The argument for the $d_L$-$d_R$ regular bipartite graph can be extended in the general case. Specifically, let $\Dl$ and $\Dr$ denote the maximum degrees of the left and right vertices in the dependency graph, respectively.
\begin{theorem}\label{thm: general baseline}
Let $X_1,\ldots,X_m\in\{0,1\}$ be $r$-wise independent random variables, and let $Y_1,\ldots,Y_n\in[0,1]$ satisfy $\E[Y_i]=p$ for all $i\in[n]$. Suppose that each $Y_i$ is a function of a subset of the $X_j$'s, and that the corresponding dependency graph has maximum left degree at most $\Dl$ and maximum right degree at most $\Dr$. We discard underlying variables that appear in no output and assume $1\le r\le m$. Assume $r\ge \Dr$. Then, for every $0<p<\alpha<1$,
$$
\Pr\left[\sum_{i=1}^n Y_i\ge \alpha n\right]
\le
\exp\left(
-\frac{\lf r/\Dr\rf}{\Dl}\D(\alpha\|p)
\right)
$$
\end{theorem}
\begin{proof}
The proof is the same as in the $d_L$-$d_R$ regular case. Indeed, if $I\subseteq[n]$ is a uniformly random subset of size $\lf r/\Dr\rf$, then the right vertices in $I$ are adjacent to at most $ r$ left vertices. Hence the corresponding underlying $X_j$'s are fully independent by $r$-wise independence. Moreover, the induced family $(Y_i)_{i\in I}$ is read-$\Dl$. Applying the read-$\Delta$ inequality and optimizing the resulting Chernoff bound gives the desired bound.
\end{proof}
The rest of the section is dedicated to improving the factor $\frac{r}{\Dl\Dr}$ in the exponent. Let $d_L$ and $d_R$ denote the average left and right degrees, so that $d_Lm = d_Rn$. 
 There are two possible sources of improvement. First, if $d_L\ll \Dl$, then a random set of $k$ right vertices should typically induce maximum left-degree much smaller than $\Dl$. Thus, as in the regular case, one may hope that the induced family is read-$d_0$ for some $d_0\ll \Dl$, leading to an improvement from $1/\Dl$ to $1/d_0$.

Second, if $d_R\ll \Dr$, then the deterministic restriction $k\le r/\Dr$ is overly pessimistic. A uniformly random set of $k$ right vertices should have neighborhood size closer to $kd_R$ left vertices, up to the appropriate correction. Hence, we may hope to take $k$ much larger than $r/\Dr$, as long as the number of vertices on the left does not exceed $r$ with high probability. This would replace the factor $\lfloor r/\Dr\rfloor$ by $k$.

Formally, let $I\subseteq[n]$ be a subset of size $k$ sampled uniformly at random, and let $L_k=|\Gamma(I)|$ and $\Lam_k=\max_{x\in[m]}|\Gamma(x)\cap I|$ denote the number of left vertices and the maximum left degree in the subgraph induced by $I$, respectively.

In the good event $\{L_k\le r\}\cap\{\Lam_k\le d_0\}$ the variables $(Y_i)_{i\in I}$ depend on at most $r$ underlying variables and form a read-$d_0$ family. Therefore, if
$
\lambda_0:=\frac1{d_0}\ln\frac{\alpha(1-p)}{p(1-\alpha)}
$,
then
$$
e^{-\lambda_0\alpha k}
\E\left[\exp\left(\lambda_0\sum_{i\in I}Y_i\right)\mid L_k\le r\cap\Lam_k\le d_0\right]
\le
\exp\left(
-\frac{k}{d_0}\D(\alpha\|p)
\right)=:G.
$$
Thus, our goal becomes to show that the two bad events ${L_k>r}$ and ${\Lam_k>d_0}$ do not contribute more than the target term $G$. Let
$$
\Pr[L_k>r]\le \veps_L,
\qquad
\Pr[\Lam_k>d_0]\le \veps_\Lam.
$$
The idea is as before: if a bad event gives a ``bad term", and $\veps$ is the probability of it occurring, then we require $$
\ln\frac1\veps
\ge
\ln\frac{\text{bad term}}{G}.
$$
This guarantees that $
\veps\cdot(\text{bad term})\le G
$, which implies
\begin{equation}\label{eq: goal}
\Pr\left[\sum_{i=1}^nY_i\ge \alpha n\right]
\le
3\exp\left(
-\frac{k}{d_0}\D(\alpha\|p)
\right).
\end{equation}

The following theorem formalizes this intuition. The subsequent corollaries give specific parameter choices. 

\begin{restatable}{theorem}{generalReadImproved}\label{thm:general-improvement}
Assume the setup of Theorem~\ref{thm: general baseline}.
Let $\chi$ be the fractional chromatic number of the graph on the underlying variables, where two vertices are adjacent if they have a common output neighbor.
Fix $0<p<\alpha<1$, and let $d_0,k$ be positive integers satisfying
$d_0\le\Dl$, $k\le n$, and $d_Rk<r$.
Suppose first that if $d_0< \Dl$, we have
\begin{equation}\label{eq:cond1}
(d_0+1)\ln\left(\frac{n}{k\Dl}\right)
+\ln((d_0+1)!)
-\ln\frac{d_Lm}{\Dl}
\ge
k\left(\frac1{d_0}-\frac1\Dl\right)\ln\frac{\alpha}{p}
\end{equation}
When $d_0=\Dl$ this condition is omitted, as in this case $\Lk\leq d_0$ with probability $1$.
Suppose also that
\begin{equation}\label{eq:cond2}
(r-d_Rk)^2
\max\left\{
\frac{1}{2k\Dr^2},
\frac{2}{m\chi}
\right\}
\ge
\left(
\frac{k}{d_0}
-
\frac{\lfloor r/\Dr\rfloor}{\Dl}
\right)
\ln\frac{\alpha}{p}.
\end{equation}
Then
$$
\Pr\left[\sum_{i=1}^nY_i\ge \alpha n\right]
\le
3\exp\left(
-\frac{k}{d_0}\D(\alpha\|p)
\right).
$$
\end{restatable}
\begin{remark}
The auxiliary graph has maximum degree at most $\Dl(\Dr-1)$, so $\chi\leq\Dl(\Dr-1)+1$.

    Let $M:=\frac{d_Lm}{\Dl}$, then~\eqref{eq:cond1} can be written as $$
    (d_0+1)\ln\left(\frac{M}{d_Rk}\right)
+\ln((d_0+1)!)
-\ln M
\ge
k\left(\frac1{d_0}-\frac1\Dl\right)\ln\frac{\alpha}{p},
    $$
    which is the same form as the condition in~\eqref{eq:exact thm3.2} of Theorem~\ref{thm:d1-d2-improved} with $m$ replaced by $M$.
\end{remark}
\begin{proof}[Proof Sketch]
The first condition~\eqref{eq:cond1} is obtained as in the proof of
Theorem~\ref{thm:d1-d2-improved}, by controlling the event
${\Lambda_k>d_0}$ and comparing its contribution with the good-event term
$G$.

For the second condition~\eqref{eq:cond2}, we control the event
${L_k>r}$. Azuma's inequality and the Janson fractional chromatic number~\cite{Janson2004LargeDF}
argument give
$$
\ln\frac1{\veps_L}
\ge
(r-d_Rk)^2
\max\left\{
\frac{1}{2k\Dr^2},
\frac{2}{m\chi}
\right\}.
$$
A direct calculation for the corresponding term $\ln\frac{\text{bad term}}{G}$ gives the required lower bound on $\ln(1/\veps_L)$. The full
proof is deferred to Appendix~\ref{app:read-d-thm-general}.
\end{proof}

Since the regular case already shows how the $\veps_\Lam$ term can be used to improve the read-degree factor from $1/\Dl$ to $1/d_0$, we first isolate the other source of improvement: replacing $r/\Dr$ by a larger value of $k$.

\begin{corollary}\label{cor:general improved dR}
Assume the setup of Theorem~\ref{thm: general baseline}. 
Let
\[
\widetilde d_R
:=
d_R+\Delta_R\sqrt{\frac{2\ln(\alpha/p)}{\Delta_L}}.
\]
Then
\[
\Pr\left[\sum_{i=1}^nY_i\ge \alpha n\right]
\le
2\exp\left(
-\frac{\lf r/{\widetilde d_R} \rf}{\Delta_L}D(\alpha\Vert p)
\right)
\]
\end{corollary}
\begin{proof}
Set $k:=\lfloor r/\widetilde d_R\rfloor$. If $\widetilde d_R\ge\Delta_R$, then
$k\le\lfloor r/\Dr\rfloor$, and Theorem~\ref{thm: general baseline} gives the desired bound.

Thus, we may assume $\widetilde d_R < \Delta_R$. Set $d_0=\Delta_L$, then we only have one bad event $L_k>r$ to account for via condition~\eqref{eq:cond2}. 

Since $k\le r/\widetilde d_R$, we have
$
r-d_Rk\ge k(\widetilde d_R-d_R)
$.
Therefore, by the definition of $\widetilde d_R$, we get
$$
\frac{(r-d_Rk)^2}{2k\Dr^2}
\geq \frac{k(\widetilde d_R-d_R)^2}{2\Dr^2}\geq\frac{k}{\Dl}\ln\frac{\alpha}{p},
$$
which implies~\eqref{eq:cond2} in Theorem~\ref{thm:general-improvement}. Hence the good-event term and the $L_k>r$ term each contribute at most 
$$
\exp\left(-\frac{k}{\Delta_L}D(\alpha\|p)\right),
$$
which gives the factor $2$.
\end{proof}

The coloring parameter can be much smaller than the maximum-degree bound. The next corollary illustrates this on a family in which most outputs depend on one underlying variable, while the remaining outputs have a common input within each component. Figure~\ref{fig:sunflower-dependency} shows one component and the corresponding auxiliary graph.

\begin{figure}[htbp]
\centering
\begin{tikzpicture}[
  x=0.9cm,y=0.9cm,line width=0.4pt,font=\small,
  input/.style={circle,draw,fill=white,inner sep=0pt,minimum size=4pt},
  output/.style={rectangle,draw,fill=black!15,inner sep=0pt,minimum size=5pt},
  colored/.style={circle,draw,fill=white,inner sep=0pt,minimum size=13pt,font=\scriptsize}
]
\node at (3.2,6.8) {(a) Dependency graph};
\node at (11.6,6.8) {(b) Input co-occurrence graph};
\node at (2.4,6.0) {private inputs};
\node at (6.0,6.0) {outputs};
\node[input,fill=black] (common) at (0.4,3.7) {};
\node[align=center] at (0.4,4.35) {common\\input};
\node[output] (out1) at (6.0,5.1) {};
\node[output] (out2) at (6.0,3.7) {};
\node[output] (out3) at (6.0,2.3) {};
\node[input] (p11) at (2.4,5.4) {};
\node[input] (p12) at (2.4,4.8) {};
\node[input] (p21) at (2.4,3.95) {};
\node[input] (p22) at (2.4,3.45) {};
\node[input] (p31) at (2.4,2.6) {};
\node[input] (p32) at (2.4,2.0) {};
\draw (common)--(out1) (common)--(out2) (common)--(out3);
\draw (p11)--(out1) (p12)--(out1);
\draw (p21)--(out2) (p22)--(out2);
\draw (p31)--(out3) (p32)--(out3);
\node at (3.2,1.35) {$D^2$ disjoint input-output pairs};
\foreach \x in {1.0,2.6,4.2} {
  \node[input] (singlein) at (\x,0.65) {};
  \node[output] (singleout) at (\x+0.6,0.65) {};
  \draw (singlein)--(singleout);
}
\node at (5.7,0.65) {$\cdots$};
\node[colored] (center) at (11.6,3.7) {1};
\node[colored] (a1) at (10.8,5.3) {2};
\node[colored] (a2) at (12.4,5.3) {3};
\node[colored] (b1) at (9.6,3.2) {2};
\node[colored] (b2) at (10.3,1.9) {3};
\node[colored] (c1) at (13.6,3.2) {2};
\node[colored] (c2) at (12.9,1.9) {3};
\draw (center)--(a1)--(a2)--(center);
\draw (center)--(b1)--(b2)--(center);
\draw (center)--(c1)--(c2)--(center);
\node at (11.6,1.2) {$D$ colors suffice};
\foreach \x in {9.6,10.3,11.0} {
  \node[input] at (\x,0.6) {};
}
\node at (11.7,0.6) {$\cdots$};
\node at (13.3,0.6) {isolated inputs};
\end{tikzpicture}
\caption{One component of the dependency graph in Corollary~\ref{cor:sunflower-comparison}, shown for $D=3$. In (a), circles are inputs and squares are outputs; the bottom row represents all $D^2$ disjoint pairs. In (b), two inputs are joined when they occur in a common output. The labels give a proper coloring of the cliques sharing the common input. For general $D$, this graph is $D$-colorable and contains a $D$-clique, so its fractional and ordinary chromatic numbers both equal $D$.}
\label{fig:sunflower-dependency}
\end{figure}

\FloatBarrier

\begin{restatable}{corollary}{sunflowerComparison}\label{cor:sunflower-comparison}
Assume the setup of Theorem~\ref{thm: general baseline}, with fixed constants $0<p<\alpha<1$. Let $D\to\infty$, and suppose that the dependency graph is a disjoint union of copies of the following graph: $D$ outputs have one common underlying variable and $D-1$ private underlying variables each; there are also $D^2$ disjoint input-output pairs. Assume $r\geq D$.
Then
\[
\Dl=\Dr=D,\qquad d_R=\frac{2D}{D+1},\qquad\chi=D.
\]
With $d_0=D$, the Azuma estimate permits $k_A=\Theta(r/\sqrt D)$ and gives
\[
\Pr\left[\sum_iY_i\geq\alpha n\right]
\leq 2\exp\left(-\Omega\left(\frac{r}{D^{3/2}}\right)\D(\alpha\|p)\right).
\]
If $r=\Theta(m)$, the Janson estimate permits $k_J=\Theta(r)$ and improves this bound to
\[
\Pr\left[\sum_iY_i\geq\alpha n\right]
\leq 2\exp\left(-\Omega\left(\frac rD\right)\D(\alpha\|p)\right).
\]
\end{restatable}

Relative to the baseline coefficient $r/D^2$, the two estimates gain factors of order $\sqrt D$ and $D$, respectively, with the latter applying when $r=\Theta(m)$. See Appendix~\ref{app:sunflower-comparison} for the proof.\\

We can also adapt Corollary~\ref{cor:regular improved dL} to the general setting.
\begin{restatable}{corollary}{generalLeftImproved}\label{cor:general improved dL}
Assume the setup of Theorem~\ref{thm: general baseline}. Let
\(
M:=\frac{d_Lm}{\Dl}\), \(
k:=\left\lfloor\frac r{\Dr}\right\rfloor
\).
Suppose $$\ln(\alpha/p)=\Theta(1),\qquad r=o(M),\qquad\text{and}\quad
k\ln\frac{M}{d_Rk}\gg \ln^2(d_Rk)
.$$
Then, regardless of $\Dl$,
\[
\Pr\left[\sum_{i=1}^nY_i\geq \alpha n\right]
\leq
2\exp\left(
-\D(\alpha\|p)\cdot
\Omega\left(
\min\left\{
k,\sqrt{k\ln\frac{M}{d_Rk}}
\right\}
\right)
\right).
\]
\end{restatable}

The proof is the same as that of Corollary~\ref{cor:regular improved dL}, with the parameter $M$ in place of $m$; we defer the details to Appendix~\ref{app:general improved dL}. In particular, the condition
$
k\ln\frac{M}{d_Rk}\gg \ln^2(d_Rk)
$
 holds when $d_R=O(1)$, analogous to the regular-case assumption. More generally, the assumption $d_R=O(1)$  in the regular case can also be replaced by the same asymptotic condition.
\\

We conclude with a simple regime in which improving both degree factors
turns a constant baseline bound into a probability that tends to zero.

\begin{restatable}{corollary}{generalBothImproved}
\label{cor:general improved both}
Assume the setup of Theorem~\ref{thm: general baseline}, with fixed
$0<p<\alpha<1$. Suppose that, as $r\to\infty$,
\[
 d_R=O(1),\qquad \Delta_R\leq\sqrt r,\qquad r\Delta_L\leq m.
\]
Then
\[
\Pr\left[\sum_{i=1}^nY_i\geq\alpha n\right]
\leq 3\exp\bigl(-\Omega(r^{1/3})\bigr).
\]
\end{restatable}

For example, when $\Delta_L=\Delta_R=\sqrt r$, the baseline exponent
$r/(\Delta_L\Delta_R)$ is only $1$. The corollary instead gives an
exponent of order $r^{1/3}$: a constant upper bound becomes a probability
that tends to zero exponentially in $r^{1/3}$.
The proof is given in Appendix~\ref{app:general improved both}.

\section{Linear Hashing with Random Inputs}\label{sec:lin hashing}
\label{sec:balls-bins-weak-randomness}
This section analyzes the balls-and-bins model in which the hash function is a random linear map over
\(\F_q\), while the input keys are not fixed adversarially.  The point is to interpolate between two
standard settings.  In the worst-case setting, one fixes an arbitrary set of keys \(B\subseteq \F_q^u\) and
then samples a random linear hash function.  In the binary case, Jaber, Kumar, and Zuckerman~\cite{10.1145/3717823.3718208} show that even in this worst-case
model the maximum load has the same order as fully random hashing, up to constants.  
For arbitrary finite fields, Kumar and Mon~\cite{Kumar2026ListDL} obtained worst-case maximum-load bounds via Furstenberg sets, matching the fully random order for every fixed $q$, with constants depending on $q$. 
Here we keep
the same linear-hashing setup, but additionally assume that the input sequence has conditional 
min-entropy.  This lets us control the probability that a sampled tuple of keys contains many linear
dependencies, which is exactly the obstruction to applying the usual product-moment upper-tail
argument.\\

The main result of this section is a fixed-bin tail bound for this semi-random-input model. Theorem~\ref{thm:main-balls} bounds the \(k\)-th product moment of the load of a fixed bin in terms of span-hit probabilities for a random tuple of input keys. Corollary~\ref{cor:optimal k} uses this estimate for pointwise block sources and gives an explicit one-bin tail bound. A union bound over the \(n\) bins then gives the corresponding maximum-load estimate, which we compare with worst-case-input linear hashing and with Bshouty's fixed-bin estimate in Table~\ref{tab:maxload}.

Throughout, let $q\geq 2$ be a prime power. Spans, ranks and linear independence are over $\F_q$.
\[
        n=q^\ell,
        \qquad U:=\F_q^u\setminus\{0\}
\]
and view \(U\) as the universe of possible keys and \(\F_q^\ell\) as the set of bins. 

Let \(\calH\) be the family of linear maps
\[
        h:\F_q^u\to \F_q^\ell.
\]
Equivalently, choose \(\ell\) independent uniform vectors
\(Z^{(1)},\ldots,Z^{(\ell)}\in\F_q^u\), and define
\[
        h(v)=\bigl(\langle v,Z^{(1)}\rangle,\ldots,
        \langle v,Z^{(\ell)}\rangle\bigr),
        \qquad v\in\F_q^u.
\]
For every nonzero \(x\in\F_q^u\), \(h(x)\) is uniform on \(\F_q^\ell\).  More generally, if
\(x_1,\ldots,x_d\) are linearly independent, then
\(h(x_1),\ldots,h(x_d)\) are mutually independent and uniform on \(\F_q^\ell\).  This last fact is the
only hash-family property used below.

Let \(R_1,\ldots,R_m\in U\) be the vectors with $u$ coordinates over $\F_q$ of the $m$ keys.  For a bin
\(y\in\F_q^\ell\), define its load and the maximum load, respectively, as
\[
        L_y:=\sum_{i=1}^m \one_{\{h(R_i)=y\}},
        \qquad
        L_{\max}:=\max_{y\in\F_q^\ell} L_y.
\]

We write
\[
\OPT(m, n) = \frac{m}{n}\,\xi,
\qquad
\xi :=
\begin{cases}
\frac{\theta}{\log \theta}, & \theta\ge 2,\\
1, & \theta<2.
\end{cases}
\]
for the order of the expected maximum load under a fully random function from $[m]$ to $[n]$,
where \(\theta(m,n) = \frac{n\log n}{m}\).

Our goal is to obtain an upper tail $\Pr[L_{\max}\ge  R\OPT(m,n)]$, or, equivalently,
$$
\Pr[L_{\max}\ge \alpha\frac{m}{n}],
$$
for $\alpha=R\xi$, given that $R>1$ and $\xi\geq 1$.  Since there are \(n\) bins, it suffices to
prove a sufficiently strong one-bin tail bound and then apply a union bound.

\subsubsection*{Span-hit probability assumptions}
 The goal of the assumptions below
is to rule out the possibility that a typical small random subcollection of $R_1\ldots, R_m$
keys has many linear dependencies, since such dependencies are precisely what
prevents random linear hashing from behaving like fully random hashing.

The cleanest sufficient assumption is a pointwise conditional min-entropy bound~\ref{def:block-source}.
However, a weaker condition, where  we only need to control span hits among a small random tuple of indices, would suffice.

\begin{definition}[Span-hit probability sequence]
\label{def:random-tuple-rank-defect-profile}
Let \(1\le k\le m\), and let \(\rho_1,\ldots,\rho_{k-1}\geq 0\).  We say that the input
sequence has span-hit probability sequence \((\rho_j)_{j=1}^{k-1}\) if the following holds.

Let \(I=\{i_1,\ldots,i_k\;|\;i_1<\ldots<i_k\}\) be a uniformly random \(k\)-subset of \([m]\), sampled independently of \(R_1,\ldots,R_m\).  Then, for every \(j=1,\ldots,k-1\),
\[
        \Pr\!\left[
        R_{i_{j+1}}\in
        \Span(R_{i_1},\ldots,R_{i_{j}})
        \,\middle|\,
        i_1,\ldots,i_k,\,
        R_{i_1},\ldots,R_{i_{j}}
        \right]
        \le \rho_jq^j.
\]
We additionally assume $\rho_jq^j\leq 1$, $j=1,\ldots,k-1$.
\end{definition}

Under Definition~\ref{def:block-source}, we may take
\[
        \rho_j=\min\{1/q^j,p_{\rm src}\}.
\]

In comparison, suppose the input is obtained by sampling without replacement
from a fixed set \(B\subseteq U\).  Then, conditionally on the previously
revealed values of the keys,
\[
        \Pr\!\left[
        R_{i_{j+1}}\in
        \Span(R_{i_1},\ldots,R_{i_{j}})
        \,\middle|\,
        R_{i_1},\ldots,R_{i_{j}}
        \right]
        \le
        \frac{
        |\Span(R_{i_1},\ldots,R_{i_{j}})\cap B|
        }{
        |B|-j
        }
        \le
        \frac{q^j-1}{|B|-j}.
\]

\subsection{A single-bin tail bound}
\begin{theorem}\label{thm:main-balls}
Fix a bin \(y\in \F_q^\ell\), and denote by
\(
        L_y:=\sum_{i=1}^m \one_{\{h(R_i)=y\}}
\).
Let \(a\) be an integer with \(1\le a\le m\), and let \(1\le k\le a\).
Assume that, for a uniformly random ordered
\(k\)-tuple of distinct indices \(i_1,\ldots,i_k\in[m]\), sampled independently
of the input sequence, we have
\[
        \Pr\!\left[
        R_{i_{j+1}}\in
        \Span(R_{i_1},\ldots,R_{i_j})
        \,\middle|\,
        i_1,\ldots,i_k,\,
        R_{i_1},\ldots,R_{i_j}
        \right]
        \le q^j\rho_j
\]
for every \(j=1,\ldots,k-1\). Then
\[
        \Pr[L_y\ge a]
        \le
        \frac{\binom{m}{k}}{\binom{a}{k}}
        \frac{1}{n^k}
        \prod_{j=1}^{k-1}
        \bigl(1+(n-1)q^j\rho_j\bigr).
\]

Consequently,
\[
        \Pr[L_{\max}\ge a]
        \le
        \min_{1\le k\le a}
        \frac{\binom{m}{k}}{\binom{a}{k}}
        \frac{1}{n^{k-1}}
        \prod_{j=1}^{k-1}
        \bigl(1+(n-1)q^j\rho_j\bigr).
\]
\end{theorem}
\begin{proof}Fix \(y\in\F_q^\ell\), and set
\(
        X_i:=\one_{\{h(R_i)=y\}}
\).
By Lemma~\ref{lem: tail reduction}, it is sufficient to upper bound
\(\E[X_{i_1}\cdots X_{i_k}]
\) for a $k$-tuple of distinct indices $1\leq i_1<\ldots<i_k\leq m$ sampled uniformly at random from $k$-subsets of $[m]$.
We now bound this expectation.

Condition on the realized vectors \(R_{i_1},\ldots,R_{i_k}\), and let
\[
        D:=\dim\Span(R_{i_1},\ldots,R_{i_k}).
\]
Choose a basis among these \(k\) vectors. If all \(k\) vectors hash to \(y\),
then in particular all vectors in this basis hash to \(y\). Since the basis
vectors are linearly independent, their images under a random linear map are
mutually independent and uniform in \(\F_q^\ell\). Therefore, using the independence of \(h\) from the input sequence,
\[
        \E_h[X_{i_1}\cdots X_{i_k}
        \mid R_1,\ldots, R_m]=\E_h[X_{i_1}\cdots X_{i_k}
        \mid R_{i_1},\ldots, R_{i_k}]
        \le n^{-D}.
\]

For \(j=1,\ldots,k\), define
\[
        D_j:=
        \one_{\{R_{i_j}\in
        \Span(R_{i_1},\ldots,R_{i_{j-1}})\}},
        \qquad D_1:=0.
\]
Then the span dimension increases exactly when \(D_j=0\), so that
\(
        D=k-\sum_{j=1}^k D_j
\).
Consequently,
\[
        n^{-D}
        =
        n^{-k}\prod_{j=1}^k n^{D_j}.
\]
By the span-hit assumption, it holds that
\[
        \Pr[D_{j+1}=1
        \mid i_1,\ldots,i_k,\,
        R_{i_1},\ldots,R_{i_j}]
        \le q^j\rho_j 
\]
for $j=1,\ldots, k-1$. 
Hence,
\[
        \E[n^{D_{j+1}}
        \mid i_1,\ldots,i_k,\,
        R_{i_1},\ldots,R_{i_j}]
        \le
        1+(n-1)q^j\rho_j ,
\]
which implies the desired bound
\[
        \E\left[\prod_{j=1}^k n^{D_j}\right]
        \le
        \prod_{j=1}^{k-1}
        \bigl(1+(n-1)q^j\rho_j\bigr).
\]
Therefore,
\[
        \Pr[L_y\geq a] 
        \le  \frac{\binom{m}{k}}{\binom{a}{k}}
        n^{-k}
        \prod_{j=1}^{k-1}
        \bigl(1+(n-1)q^j\rho_j\bigr).
\]
The maximum-load bound follows by a union bound over the \(n\) bins.
\end{proof}

\begin{restatable}{corollary}{optimalLoadBound}
\label{cor:optimal k}
Let the assumptions of Theorem~\ref{thm:main-balls} hold with $\rho_j\leq p_{\rm src}$. Fix $\alpha>1$, and set $a := \left\lceil\alpha m/n\right\rceil$. Assume $1\le a\le m$, $\alpha m/n\geq 1$ and
\[
        1\le \log_q\frac{a}{mp_{\rm src}}
        \qquad\text{and}\qquad
        \left(\log_q\frac{a}{mp_{\rm src}}\right)^2\le a .
\]
Then, for an absolute constant \(C>0\),
\[
        \Pr[L_y\ge \alpha\frac{m}{n}]
        \le
        C\alpha^C
        (np_{\rm src})^{\log_q\alpha}
        q^{-\frac12\log_q^2\alpha}.
\]
Consequently,
\[
        \Pr[L_{\max}\ge \alpha \frac{m}{n}]
        \le
        Cn\alpha^C
        (np_{\rm src})^{\log_q\alpha}
        q^{-\frac12\log_q^2\alpha}.
\]
\end{restatable}
\begin{proof}[Proof Sketch]
    The proof is an optimization of the bound from Theorem~\ref{thm:main-balls}.
With $\rho_j\leq p_{\rm src}$, the relevant expression is
\[
A_k:=
\frac{\binom{m}{k}}{\binom{a}{k}}\frac1{n^k}
\prod_{j=1}^{k-1}(1+(n-1)p_{\rm src}q^j).
\]
We set
\[
k:=\left\lf\log_q\frac{a}{mp_{\rm src}}\right\rf
=
\log_q\frac1{np_{\rm src}}+\log_q\alpha+O(1).
\]
The binomial term contributes at most $ e\alpha^{-k}
        \leq
        e\alpha
        (np_{\rm src})^{\log_q\alpha}
        q^{-\log_q^2\alpha}$, while the product term
contributes at most $C_0\alpha
q^{\frac12\log_q^2\alpha}$. Combining these two
estimates gives the stated bound. The details are in
\mbox{Appendix~\ref{app:max-load-optimization}}.
\end{proof}

\begin{table}[t]
\centering
\small
\begin{tabular}{c|c|c}
\textbf{Reference / field} & \textbf{Method} & \textbf{Maximum-load tail at \(T_R\)}\\
\hline
Ideal hashing / any \(q\) & independent balls & \(n^{1-R+o(1)}\)\\[1.2ex]
\cite{10.1145/3717823.3718208} / \(q=2\)
& potential method & \(O\!\left(\frac{(\log n/\log\log n)^2}{T_R^2}\right)\)\\[1.2ex]
\cite{bshouty2026notesecondorderexpectedmaximumload} / \(q=2\)
& optimized potential & \(O\!\left(\frac{(\log n)^{2/R}}{T_R^2}\right)\)\\[1.2ex]
\cite{Kumar2026ListDL} / any \(q\)
& Furstenberg sets & \(O\!\left(\frac{(\log n)^{8q/(R\ln q)}}{T_R}\right)\)\\[1.2ex]
\cite{bshouty2026notesecondorderexpectedmaximumload} (extended) / any \(q\)
& fixed bin + union bound & \(nq^{-\beta_R^2+O(\beta_R)}\)\\[1.2ex]
This work / any \(q\)
& fixed bin + union bound & \(n(np_{\rm src})^{\beta_R}q^{-\beta_R^2/2+O(\beta_R)}\)
\end{tabular}
\caption{Comparison of maximum-load upper-tail bounds in the balanced case $m=n=q^\ell$, where
$T_R = R\frac{\log n}{\log\log n}$,
\(\beta_R=\log_q T_R\). The fixed-bin counting bound of
\cite{bshouty2026notesecondorderexpectedmaximumload}
was originally stated for $q=2$; its argument extends directly to
$\mathbb F_q$ for any prime power $q$.}
\label{tab:maxload}
\end{table}

Applying the preceding fixed-bin estimate and then taking a union bound over the \(n=q^\ell\) bins gives the following maximum-load comparison in the balanced case (see Table~\ref{tab:maxload}). 
\section{Stochastic Processes and Random Walks}\label{sec:RWs}
This section applies the product-moment framework to concentration for dependent stochastic processes, with Markov chains as the main example. Theorem~\ref{thm: spectral bound} gives a fixed-index product bound for Markov-chain observables in terms of the gaps between the sampled times. Averaging this bound over a uniformly random \(k\)-subset and using negative association of the gaps gives a spectral tail bound of the usual scale. We then use the same viewpoint to recover a mixing-time-based bound: Theorem~\ref{thm: mixing time} applies the spectral product bound to the \(T\)-step chain \(M^T\), where \(T\) is a total-variation mixing time, and obtains a Chernoff-type upper tail depending on \(t/T\). In the small-deviation regime, this yields an exponent of order \(\delta^2\mu t/((1-\mu)T)\).
\\

Let $X_1,\ldots,X_t$ be random variables taking values in $[0,1]$ with $\E[X_i]=\mu$, $i=1,\ldots, t$, and let
$
X:=\sum_{i=1}^t X_i
$.
The basic principle is that concentration can be obtained from upper bounds on products
$
\E[X_{i_1}\cdots X_{i_k}]
$
for fixed tuples of indices
$
1\leq i_1<\cdots<i_k\leq t
$.
In particular, consider the following abstract framework. Suppose that---for all $k$ or only for those $k$ of interest---we have for every fixed tuple $i_1<\cdots<i_k$ a product-moment bound of the form
\begin{equation}
\label{eqn:corr-decay1}
\E[X_{i_1}\cdots X_{i_k}]
\leq
c\prod_{j=1}^{k-1}\psi_\mu(i_{j+1}-i_j),
\end{equation}
where $\psi_\mu$ is a non-increasing function of the gap between consecutive sampled indices, and $c$ depends only on the initial distribution. This type of estimate says that dependence between sampled variables decays as the sampled times $i_j$ become farther apart. The exact form of $\psi_\mu$ depends on the process under consideration. 
\iffalse ***
Note in particular that when $X_i \in \{0,1\}$, the following "correlation decay on the line" bound suffices for (\ref{eqn:corr-decay1}) to hold---where we can take $C = c$ if $c \geq 1$: 
\begin{equation}
\label{eqn:corr-decay1}
\forall j~(1 \leq j \leq k-1), 
\E[X_{i_{j+1}} \bigm| X_{i_1} = X_{i_2} = \cdots = X_{i_j} = 1]
\leq
C \ddot \psi_\mu(i_{j+1}-i_j). 
\end{equation}
*** \fi 

Now, in order to bound $\E[X_{i_1}\cdots X_{i_k}]$ for a $k$-tuple $(i_1,\ldots,i_k)$ sampled uniformly at random, we show that the gaps
$$
g_0=i_1,
\qquad
g_j=i_{j+1}-i_j
\quad\text{for }j=1,\ldots,k-1,
$$
 as functions of a uniformly sampled $k$-subset $(i_1,\ldots, i_k)$, are negatively associated. Consequently, for non-increasing $\psi_\mu$,
$$
\E\left[\prod_{j=1}^{k-1}\psi_\mu(g_j)\right]
\leq
\prod_{j=1}^{k-1}\E[\psi_\mu(g_j)].
$$
Thus, if
$$
p_k:=\E[\psi_\mu(G)]
$$
for a gap $G$ with the relevant marginal distribution, then the product-moment hypothesis implies
$$
\E S_k(X_1,\ldots,X_t)
\leq
c\binom{t}{k}p_k^{k-1}.
$$
This, in turn, gives an upper tail estimate for \(\Pr[X\geq \alpha t]\). The final step is to optimize this finite-degree bound over $k$.\\

We note that the abstract framework above applies whenever one can establish suitable product-moment bounds for fixed tuples. For example, in the binary case $X_i\in\{0,1\}$ the inequality follows from a \emph{correlation decay} guarantee of the form
$$
\Pr\left[
X_{i_j}=1
\mid
X_{i_1}=\cdots=X_{i_{j-1}}=1
\right]
\leq
\psi_\mu(i_j-i_{j-1}), \qquad j=2,\ldots,k.$$

In our main application below, we assume that the stochastic process is a Markov chain and recover an asymptotically optimal tail bound.

Let
$
X_i=f_i(v_i)
$, 
where $(v_1,\ldots,v_t)$ is a walk on a finite-state Markov chain on a state space $[n]$ and $f_i:[n]\to[0,1]$ have common stationary mean $\E_\pi[f_i(v_i)]=\mu$. In this setting, Theorem~\ref{thm: spectral bound} gives the concrete choice
$$
\psi_\mu(g)=\mu+\lambda(M)^g(1-\mu).
$$
We then use the blocking idea of \cite{Chung2012ChernoffHoeffdingBF} to obtain a bound depending only on the total-variation mixing time. 

\subsection{Spectral product bounds for Markov chains}
In this subsection, using standard linear algebra, we prove the fixed-tuple product bound for an arbitrary ergodic finite-state Markov chain.
\begin{restatable}{theorem}{spectralProductBound}\label{thm: spectral bound}
Let $M$ be an ergodic Markov chain on $[n]$ with stationary distribution $\pi$. $M$ also denotes the transition matrix. Let \(\lambda:=\lambda(M)\) be the \(L_2(\pi)\)
spectral norm of \(M\).
Let $\varphi$ be the initial distribution of $v_1$ on $[n]$. Fix
$
1\leq i_1<\cdots<i_k\leq t$,
and let $g_j$, $j\in[k-1]$ be defined as before.
Then
$$
\E_\varphi\left[\prod_{j=1}^k f_{i_j}(v_{i_j})\right]
\leq
\|\ph\|_{\pi^{-1}}\sqrt \mu
\prod_{j=1}^{k-1}
\left(\mu+\lambda^{g_j}(1-\mu)\right),
$$
 where $\|\ph\|_{\pi^{-1}}^2 = \sum_{v\in[n]}\frac{\varphi^2(v)}{\pi(v)}$.
Moreover,
\iffalse
$$
\E_\varphi\left[\prod_{j=1}^k f_{i_j}(v_{i_j})\right]
\leq
\left(
p+\sqrt p\,\lambda^{g_0}\|h-1\|_{2,\pi}
\right)
\prod_{j=1}^{k-1}
\left(p+(1-p)\lambda^{g_j}\right).
$$
\fi
if $\varphi=\pi$, then
$$
\E_\pi\left[\prod_{j=1}^k f_{i_j}(v_{i_j})\right]
\leq
\mu\prod_{j=1}^{k-1}
\left(\mu+\lambda^{g_j}(1-\mu)\right).
$$
\end{restatable}

\begin{proof}[Proof Sketch]
Write $P_i$ for the diagonal operator with entries $f_i$ and set
$H_i=P_i^{1/2}$. By the Markov property, the fixed-tuple product moment can be
written as a product of operators
\[
\langle h,M^{i_1-1}P_{i_1}M^{g_1}P_{i_2}\cdots M^{g_{k-1}}P_{i_k}1\rangle_\pi,
\]
where $h=d\varphi/d\pi$. After inserting the square roots $H_i$, the problem
reduces to bounding
\[
\|H_iM^gH_j\|_{2,\pi}.
\]
The orthogonal decomposition into constants and mean-zero functions, together
with the definition of \(\lambda(M)\), gives
\[
\|H_iM^gH_j\|_{2,\pi}\leq \mu+\lambda^g(1-\mu).
\]
Multiplying these bounds over the gaps gives the claimed product-moment
estimate. The full proof is given in Appendix~\ref{app:spectral-product-bound}.
\end{proof}

To deal with the expectation over a random choice of $(i_1,\ldots, i_k)$, we prove an elementary fact about the gaps induced by a uniformly sampled $k$-subset of $[t]$.

\begin{lemma}[Gap distribution for a random $k$-subset]
Let $g_0,\ldots, g_{k-1}$ be defined as before, with $g_k = t+1-i_k$.
Then $g_0,\ldots,g_{k}$ are negatively associated.
\end{lemma}
\begin{proof}
The tuple \(i_1<\cdots<i_k\) is uniquely determined by the gap vector \(g_0,\ldots,g_k\), since \(i_1=g_0\) and \(i_j=g_0+\cdots+g_{j-1}\). The gaps satisfy \(g_j\geq1\) and
\[
    \sum_{j=0}^k g_j=t+1.
\]
Therefore, the gap vector $(g_j)_{0\leq j\leq k}$ is uniform over all positive partitions of \(t+1\) into \(k+1\) parts.

Let \(g'_0,\ldots,g'_k\) be i.i.d. geometric random variables with \(\Pr[g'_j=s]=(1-q)q^{s-1}\), $s\in\mathbb Z_{+}$ for some \(q\in (0,1)\). The geometric probability mass function is log-concave, since \(p_s^2=p_{s-1}p_{s+1}\) on its support, which is the interval $\mathbb Z_+$. Thus, by Lemma~\ref{lemma: NA} $ (g'_0,\ldots,g'_k)\mid \sum_jg'_j=t+1$ are negatively associated.

For any \(y_0,\ldots,y_k\in\mathbb Z_+\) with \(\sum_j y_j=t+1\),
\[
    \Pr[g'_0=y_0,\ldots,g'_k=y_k]
    =
    (1-q)^{k+1}q^{\sum_j y_j-k-1}
    =
    (1-q)^{k+1}q^{t-k},
\]
which is constant when \(\sum_jg'_j=t+1\) is fixed. Hence,
\(    (g'_0,\ldots,g'_k)\mid \sum_jg'_j=t+1
\)
is uniform over all positive partitions of \(t+1\) into \(k+1\) parts and, therefore, has the same distribution as \((g_0,\ldots,g_k)\).  Hence, \(g_0,\ldots,g_k\) are negatively associated. 
\end{proof}

Next, we average the fixed-index product bound over a uniformly chosen $k$-subset of $[t]$. By Theorem~\ref{thm: spectral bound} for every fixed $k$-subset $\{i_1,\ldots, i_k\}$ with $1\leq i_1<\ldots<i_k\leq t$ we have
$$
\E_\varphi\left[\prod_{j=1}^k f_{i_j}(v_{i_j})\right]
\leq
c(\varphi)\cdot
\prod_{j=1}^{k-1}
\left(\mu+\lambda^{g_j}(1-\mu)\right),
$$
where $c(\varphi)=\begin{cases}
    \mu &\varphi =\pi,\\
    \sqrt{\mu}\|\varphi\|_{\pi^{-1}}&\text{for an arbitrary } \varphi.
\end{cases}$
\\
The random variables $g_0,\ldots, g_k$ are negatively associated and have the same marginal distribution. Since the map $g\mapsto \mu+\lam^g(1-\mu)$ is non-increasing, negative association gives
$$
\E_{I}\prod_{j=1}^{k-1} (\mu+\lam^{g_j}(1-\mu))\leq 
\left(\mu+(1-\mu)\E[\lambda^{g_0^{(k)}}]
\right)^{k-1}
$$
Consequently,
$$
\E S_k(f_1(v_1),\ldots,f_t(v_t))
\leq
c(\varphi)\binom{t}{k}
\left(
\mu+(1-\mu)\E[\lambda^{g_0^{(k)}}]
\right)^{k-1},
$$
which then implies concentration.
We now estimate the
marginal expectation \(\mathbb E[\lambda^{g_0}]\), given $g_0=i_1$.

First, note that \[\Pr[i_1=s] = \binom{t-s}{k-1}/\binom{t}{k} = \frac{k}{t}\prod_{\ell= 1}^{s-1}\lb1-\frac{k-1}{t-\ell}\rb\leq \frac{k}{t},\qquad s=1,\ldots, t-k+1\]
Thus,
$$
\E[\lambda^{i_1}]=
\sum_{s=1}^{t-k+1} \lambda^s\Pr[g_0=s]
\leq
\frac{k}{t}\sum_{s\geq1}\lambda^s=
\frac{\lambda}{1-\lambda}\frac{k}{t}.
$$
In other words, we set
$$
p_k(\lambda):=
\mu+(1-\mu)\frac{\lambda k/t}{1-\lambda},
$$

so that $$
\Pr[X\geq \alpha t]
\leq
c(\varphi)\frac{\binom{t}{k}}{\binom{\alpha t}{k}}
p_k(\lam)^{k-1}.
$$

\subsection{A mixing-time bound}
We now derive a concentration bound depending only on the mixing time of the chain. This setting is of interest when the spectral norm of the original transition matrix is not known or is uninformative. 
\begin{restatable}{theorem}{mixingTimeBound}\label{thm: mixing time}
Let $M$ be an ergodic finite-state Markov chain with stationary distribution $\pi$, and let $T=T(\varepsilon)$ be its total-variation mixing time for some fixed $\varepsilon\leq 1/8$. Let $(v_1,\ldots,v_t)$ be a $t$-step walk, and assume for simplicity that $T$ divides $t$.
For each $i\in[t]$, let $f_i:[n]\to[0,1]$ satisfy
$
\E_\pi f_i=\mu.
$
Set
$
X:=\sum_{i=1}^t f_i(v_i)
$, $t_0:=t/T$.

Then for every $0<\mu<\alpha<1$ and every integer $1\leq k\leq \lfloor \alpha t_0\rfloor$,
$$
\Pr[X\geq \alpha t]
\leq
c(\varphi)
\frac{\binom{t_0}{k}}{\binom{\lfloor \alpha t_0\rfloor}{k}}
p_k^{k-1},
$$
where
$
p_k
=
\mu+(1-\mu)\frac{k}{t_0}
$ and $c(\varphi)=\begin{cases}
    \mu &\text{if }\varphi =\pi;\\
    \sqrt{\mu}\|\varphi\|_{\pi^{-1}}&\text{for arbitrary } \varphi.
\end{cases}$

In particular, if $$\alpha=(1+\delta)\mu, \qquad\delta=o(1),\qquad
\frac{\mu\delta}{1-\mu}=o(1)
\qquad\text{and}\quad
 \frac{\mu\delta^2}{1-\mu}t_0\to\infty,$$
then choosing
$$
k=
\left\lfloor
\frac{\mu\delta t_0}{3(1-\mu)}
\right\rfloor
$$
for large enough $t$, so that $k\geq 1$, gives
$$
\Pr[X\geq (1+\delta)\mu t]
\leq
c(\varphi)
\exp\left(
-\frac{\mu\delta^2}{6(1-\mu)}\frac{t}{T}
+
o\left(\frac{\mu\delta^2}{1-\mu}\frac{t}{T}\right)
\right).
$$
\end{restatable}
Thus the product-moment argument recovers the mixing-time concentration scale
$$
\exp\left(
-\Theta\left(\frac{\mu\delta^2t}{T}\right)
\right).
$$

\begin{proof}
Since $\epsilon\leq1/8$, by Lemma~\ref{lem: M^T} we have
$$
\lambda(M^T)\leq \frac12.
$$ 
Let $t_0:=t/T$ and partition the walk into subwalks with transition matrix $M^T$ $$X^{(r)}:=\sum_{\ell=0}^{t_0-1} f_{r+\ell T}(v_{r+\ell T}).$$
Then
$
X=\sum_{r=1}^T X^{(r)}.
$
Let $S_k^{(r)}:=S_k((f_{r+\ell T}(v_{r+\ell T}))_{\ell=0,\ldots,t_0-1})$, which coincides with $\binom{X^{(r)}}{k}$ for indicators.
Then following \cite{Chung2012ChernoffHoeffdingBF}, we use $\frac{1}{T}\sum_{r=1}^TS_k^{(r)}$ as an analogue of their exponential moment.
In particular, we claim that
\[
\Pr[X\geq \alpha t]\leq \frac{1}{T}\sum_{r=1}^T\frac{\E[\Skr]}{\binom{\lf\alpha t_0\rf}{k}}
\]
Indeed, in the case $f_i\in\{0,1\}$, this is simply because $\frac{X}{T}=\frac{1}{T}\sum X^{(r)}\geq \alpha t_0$ implies $$
\binom{\alpha t_0}{k}\leq \binom{\frac{1}{T}\sum X^{(r)}}{k}\leq \frac{1}{T}\sum_{r=1}^T\binom{X^{(r)}}{k},
$$
where the second inequality is by Jensen, since $\binom{x}{k}$ is convex in $x$. If $f_i\in[0,1]$, then by Lemma~\ref{lem: z-cont} we have
\[
\frac1T\sum_{r=1}^T S_k^{(r)}
\ge
\frac1T\sum_{r=1}^T
\binom{X^{(r)}}{k}_{\mathrm{lin}}
\]
and by convexity and monotonicity
\[\frac1T\sum_{r=1}^T
\binom{X^{(r)}}{k}_{\mathrm{lin}}
\ge
\binom{\frac1T\sum_{r=1}^T X^{(r)}}{k}_{\mathrm{lin}}\geq \binom{\left\lf \frac1T\sum_{r=1}^T X^{(r)}\right\rf}{k}\geq \binom{\alpha t_0}{k}
\]

Thus, it remains to bound $\E[\Skr]$. Since
\(\widetilde M\) is an \(L_2(\pi)\)-contraction, 
\[\|\varphi M^{r-1}\|_{\pi^{-1}}=\| \widetilde M^{r-1}h\|_{2,\pi}\le\|h\|_{2,\pi} =\|\varphi\|_{\pi^{-1}},\] and by Theorem~\ref{thm: spectral bound} we have \[
\E[\Skr]\leq c(\varphi)\cdot\binom{t_0}{k}\lb\mu + (1-\mu)\frac{k}{t_0}\rb^{k-1}=c(\varphi)\binom{t_0}{k}p_k^{k-1}.
\]

Thus, ignoring floors, we obtain
$$
\Pr[X\geq \alpha t]
\leq
c(\varphi)
\frac{\binom{t_0}{k}}{\binom{\alpha t_0}{k}}
p_k^{k-1}.
$$
The stated Chernoff-style bound follows by optimizing this expression over $k$; the calculation is deferred to Appendix~\ref{app:mixing-optimization}.
\end{proof}

% STOC 2027 CFP, "Use of generative AI and LLMs":
% https://acm-stoc.org/stoc2027/stoc2027-cfp.html
% The AI methodology statement belongs at the end of the paper.
\section*{AI methodology}
\label{sec:ai-methodology}
OpenAI Codex assisted in discovering and formulating
Corollaries~\ref{cor:sunflower-comparison} and~\ref{cor:general improved both}
and in developing their proofs. The authors verified the correctness and originality of all content.

\printbibliography
\appendix
\section{Deferred Proofs of Section~\ref{sec:read}}
\subsection{Proof of Theorem~\ref{thm:d1-d2-improved}}\label{app:read-d-thm-regular}

\regularReadImproved*

The proof idea is that if the number of sampled right vertices $k$ is small enough, then the maximum left-degree in the induced subgraph is typically much smaller than $d_L$.

Formally, let $I=\{i_1,\ldots,i_k\}$ be a uniformly random $k$-subset of $[n]$, where $k=\lfloor r/d_R\rfloor$, so that the vertices in $I$ are adjacent to at most $kd_R\leq r$ left variables, which are independent. Let $\Lambda_k$ denote the maximum left-degree in the subgraph induced by $I$, and assume that for some $\veps_\Lambda\in (0,1)$ we have
$$
\Pr[\Lambda_k>d_0]\leq \veps_\Lambda.
$$
If $\Lambda_k\le d_0$, then the induced family is read-$d_0$, otherwise, it is still read-$d_L$. By Lemmas~\ref{lem:sym-exp} and~\ref{lem: read-Delta}
$$
\E\left[\exp\left(\lambda\sum_{j=1}^kY_{i_j}\right)\right]\leq (1-\veps_\Lambda)\left(1-p+pe^{\lambda d_0}\right)^{k/d_0}+\veps_\Lambda\left(1-p+pe^{\lambda d_L}\right)^{k/d_L}.
$$
Applying Lemma~\ref{lem:sym-exp} gives
\begin{equation}\label{eq: tail-bound1}
\Pr[Y\ge \alpha n]\leq (1-\veps_\Lambda)e^{-kF_{d_0}(\lambda)}+\veps_\Lambda e^{-kF_{d_L}(\lambda)},
\end{equation}
where
$$
F_d(\lambda):=\lambda\alpha-\frac{1}{d}\ln(1-p+pe^{\lambda d}).
$$
We evaluate this bound at
\(
\lambda_0=\frac{1}{d_0}\ln\frac{\alpha(1-p)}{p(1-\alpha)}.
\)
Note that
\(F_{d_0}(\lambda_0)=\frac{1}{d_0}\D(\alpha\|p).
\)
It remains to show that the first term dominates the tail bound in~\eqref{eq: tail-bound1}, for which it would suffice to prove that
$$
\ln\frac{1}{\veps_\Lambda}\geq k\left(F_{d_0}(\lambda_0)-F_{d_L}(\lambda_0)\right).
$$
We first bound the gap between the two exponents. Let
\(
z:=\frac{\alpha(1-p)}{p(1-\alpha)}
\), so that
$$
F_{d_0}(\lambda_0)-F_{d_L}(\lambda_0)=\frac{1}{d_L}\ln(1-p+pz^{d_L/d_0})-\frac{1}{d_0}\ln(1-p+pz).
$$
Since $z>1$ and $d_L\ge d_0$, we have
$
1-p+pz^{d_L/d_0}\leq z^{d_L/d_0-1}(1-p+pz)
$. This implies that
$$
F_{d_0}(\lambda_0)-F_{d_L}(\lambda_0)\leq \left(\frac{1}{d_0}-\frac{1}{d_L}\right)\left(\ln z-\ln(1-p+pz)\right)=\left(\frac{1}{d_0}-\frac{1}{d_L}\right)\ln\frac{\alpha}{p}.
$$
Now, to estimate $\veps_\Lambda$, we use the union bound over the bad events that a left-vertex is adjacent to more than $d_0$ vertices in $I$. For a fixed left vertex $x$, if $|N(x)\cap I|\ge d_0+1$, then some $(d_0+1)$-subset of its $d_L$ neighbors is contained in $I$. Thus,
$$
\Pr[|N(x)\cap I|\ge d_0+1]\leq \binom{d_L}{d_0+1}\left(\frac{k}{n}\right)^{d_0+1}.
$$
Therefore, by the union bound, we obtain
$$
\veps_\Lambda\leq m\binom{d_L}{d_0+1}\left(\frac{k}{n}\right)^{d_0+1}\leq
 \frac{m}{(d_0+1)!}\lb\frac{d_Lk}{n}\rb^{d_0+1}=\frac{m}{(d_0+1)!}\lb\frac{d_Rk}{m}\rb^{d_0+1}
$$
since $d_Lm=d_Rn$. Since $d_Rk\leq r$, we obtain
\[
\ln\frac1{\veps_\Lambda}
\geq(d_0+1)\ln\frac{m}{d_Rk}-\ln m+\ln((d_0+1)!)
\geq d_0\ln\frac mr-\ln r+\ln((d_0+1)!).
\]
By the assumption of the theorem, we conclude that
$$
\ln\frac{1}{\veps_\Lambda}\geq k\left(\frac{1}{d_0}-\frac{1}{d_L}\right)\ln\frac{\alpha}{p}\geq k\left(F_{d_0}(\lambda_0)-F_{d_L}(\lambda_0)\right).
$$
Hence, the bad term is at most the good term, and so we obtain the stated bound.
$$
\Pr[Y\ge \alpha n]\leq 2e^{-kF_{d_0}(\lambda_0)}=2\exp\left(-\frac{k}{d_0}\D(\alpha\|p)\right).
$$
\subsection{Proof of Theorem~\ref{thm:general-improvement}}\label{app:read-d-thm-general}

\generalReadImproved*

One of the estimates in the proof of the main theorem uses a concentration inequality based on fractional chromatic number~\cite{Janson2004LargeDF}. Specifically, the variant we use is proved by the same argument, but the independence assumption inside each cover class is weakened to negative association.

\begin{lemma}
\label{lem:janson-na}
Let $(Z_x)_{x\in V}$ be random variables taking values in $[0,1]$, and set
$
Z:=\sum_{x\in V} Z_x 
$.
Suppose that $\{(V_j,w_j)\}_j$ is a fractional cover of $V$, i.e.
$V_j\subseteq V$, $w_j\geq 0$, and
$$
\sum_{j:x\in V_j} w_j\geq 1
\qquad\text{for every }x\in V.
$$
Assume moreover that, for each $j$, the variables $(Z_x)_{x\in V_j}$ are
negatively associated. Let
$$
W:=\sum_j w_j .
$$
Then, for every $t>0$,
$$
\Pr[Z\geq \E Z+t]
\leq
\exp\left(-\frac{2t^2}{W|V|}\right).
$$
\end{lemma}

\begin{proof}[Proof Sketch]
This follows from the same proof as Janson's fractional-cover Hoeffding
bound~\cite{Janson2004LargeDF}. The only point where independence
inside each cover class is used is to factor the exponential moment over that
class. Namely, for a cover class $V_j$ and $\lambda>0$, independence gives
$$
\E\exp\left(\lambda\sum_{x\in V_j} Z_x\right)
=
\E\prod_{x\in V_j} e^{\lambda Z_x}
=
\prod_{x\in V_j}\E e^{\lambda Z_x}.
$$
Under negative association, the equality is replaced by the inequality
$$
\E\exp\left(\lambda\sum_{x\in V_j} Z_x\right)
=
\E\prod_{x\in V_j} e^{\lambda Z_x}
\leq
\prod_{x\in V_j}\E e^{\lambda Z_x},
$$
since each function $z\mapsto e^{\lambda z}$ is increasing. This is all that
is needed for the upper-tail argument, and the rest of Janson's proof is
unchanged.
\end{proof}

We next prove the technical lemma that immediately implies Theorem~\ref{thm:general-improvement}.

\begin{lemma}\label{lem:read-technical}
\begin{enumerate}[label=(\roman*)]
\item In order to obtain~\eqref{eq: goal}
it suffices to have
$$
\ln\frac1{\veps_\Lam} \ge k\left(\frac1{d_0}-\frac1\Dl\right)\ln\frac{\alpha}{p}
\qquad\text{and}\qquad
\ln\frac1{\veps_L} \ge
\left[\left(
\frac{k}{d_0}-\frac{\lf r/\Dr\rf}{\Dl} \right)\ln\frac{\alpha}{p}
\right]_+.
$$
\item
Moreover, the following estimates may be used. First,
\begin{equation}\label{eq: eps_Lambda}
\ln\frac{1}{\veps_\Lam}\geq(d_0+1)\ln\lb\frac{n}{k\Dl}\rb +\ln((d_0+1)!) - \ln\frac{d_Lm}{\Dl}.
\end{equation}
Second, if $r>d_Rk$, then
we may take
$$
\ln\frac1{\veps_L}\ge(r-d_Rk)^2\max\left\{\frac{1}{2k\Dr^2},\frac{2}{m\chi}\right\}.
$$
\end{enumerate}
\end{lemma}
\begin{proof}
    \begin{enumerate}[label=(\roman*)]
    \item 
    The first condition is the same as in Theorem~\ref{thm:d1-d2-improved}. We prove the second condition analogously. 

    Note that when $k\leq \lf r/\Dr\rf$, the event $L_k>r$ is empty. Otherwise, set $s: = \lf r/\Dr\rf$ and $d:=\Dl$ as the fallback parameters and argue as follows. 
    The corresponding bad term is at most
    $$
e^{\lambda(k-s-\alpha k)}
(1-p+pe^{\lambda d})^{s/d},
$$
since $s$ right vertices depend on at most $s\Dr\leq r$ underlying variables, while the remaining $k-s$ variables are bounded by $1$.
Thus, we have
    \[
\ln\frac{\text{bad term}}{G}
=
\frac{k}{d_0}\D(\alpha\|p) + \lam(k-s-\alpha k) + \frac{s}{d}\ln(1-p+pe^{\lam d}).
\]
Since $D(\alpha\|p)+(1-\alpha)\ln z=
\ln\frac{\alpha}{p}
$ and we use $\lambda_0=\frac{1}{d_0}\ln z$ for $
z:=\frac{\alpha(1-p)}{p(1-\alpha)}$, we obtain
\begin{align*}
&\frac{k}{d_0}D(\alpha\|p)
-\lambda_0(\alpha k-k+s)
+
\frac{s}{d}\ln(1-p+pe^{\lambda_0 d})
\\
&=
\frac{k}{d_0}\left(D(\alpha\|p)+(1-\alpha)\ln z\right)
-
\frac{s}{d_0}\ln z
+
\frac{s}{d}\ln(1-p+pz^{d/d_0})
\\
&=
\frac{k}{d_0}\ln\frac{\alpha}{p}
+
s\left[
\frac{1}{d}\ln(1-p+pz^{d/d_0})
-
\frac{1}{d_0}\ln z
\right].
\end{align*}
Next, we bound the square bracket. Since
$
1-p+pz^{d/d_0}
\le
z^{d/d_0-1}(1-p+pz)
$,
we get
$$
\frac{1}{d}\ln(1-p+pz^{d/d_0})
\le
\left(\frac{1}{d_0}-\frac{1}{d}\right)\ln z
+
\frac{1}{d}\ln(1-p+pz).
$$
Therefore, our expression becomes
$$
\frac{k}{d_0}\ln\frac{\alpha}{p}
-\frac{s}{d}\ln\frac{z}{1-p+pz} =
\left(
\frac{k}{d_0}
-
\frac{s}{d}
\right)
\ln\frac{\alpha}{p},
$$
where the last equality is because
$$
1-p+pz
=
1-p+p\frac{\alpha(1-p)}{p(1-\alpha)}
=
(1-p)\left(1+\frac{\alpha}{1-\alpha}\right)
=
\frac{1-p}{1-\alpha}.
$$
We take the positive part as $\veps_L\leq 1$.

\item First, we bound $\veps_\Lam$. 
Note that if for some left vertex $j$ more than $d_0$ of its neighbors are in $I$, it is equivalent to saying that some $(d_0+1)$-subset of $N(j)$ is contained in $I$. Thus, since each right vertex has probability $k/n$ of being in $I$ and $\deg(j)\leq \Dl$, we have \begin{align*}
\Pr[\Lambda_k>d_0] &\leq \sum_{j=1}^m\binom{\deg(j)}{d_0+1}\frac{(k)_{d_0+1}}{(n)_{d_0+1}}\leq \frac{(k/n)^{d_0+1}}{(d_0+1)!}\sum_{j=1}^m\deg(j)^{d_0+1}\\
&\leq \frac{(k/n)^{d_0+1}}{(d_0+1)!} md_L\Dl^{d_0} = \lb\frac{k\Dl}{n}\rb^{d_0+1}\frac{md_L/\Dl}{(d_0+1)!}
.
\end{align*}
If this is our choice of $\veps_\Lam$, then~\eqref{eq: eps_Lambda} follows.

Second, we bound $\veps_L$. To show that $\ln\frac{1}{\veps_L}\geq \frac{(r-d_Rk)^2}{2k\Dr^2}$, define the Doob martingale
\[
M_\ell:= \E[L_k\mid i_1,\ldots,i_\ell],
\qquad \ell=0,\ldots,k ,
\]
where $I={i_1,\ldots,i_k}$ is exposed sequentially. Revealing one right vertex can change the neighborhood size on the left by at most $\Dr$, so \(
|M_\ell-M_{\ell-1}| \le \Delta_R
\). Then by Azuma's inequality,
\[
\Pr[L_k - \E L_k > \varepsilon]
\le
\exp\!\left(
-\frac{\varepsilon^2}{2k\Delta_R^2}
\right).
\]

Therefore,
\[
\Pr[L_{k}>r]
\le
\exp\!\left(
-\frac{(r-\E[L_k])_+^2}{2k\Delta_R^2}
\right).
\]

For the second tail estimate, we follow Janson's approach~\cite{Janson2004LargeDF} to obtain
\begin{equation}\label{eq: Janson}
\Pr[L_k>r]
\le
\exp\left(
-\frac{2(r-\E L_k)_+^2}{\chi m}
\right),
\end{equation}
where $\chi$ is the fractional chromatic number of the auxiliary graph in Theorem~\ref{thm:general-improvement}.
Specifically, we use Lemma~\ref{lem:janson-na} with $V:=[m]$ and $Z_j=\bone[N(j)\cap I\neq \emptyset]$, $j=1,\ldots, m$, so that $L_k = Z=\sum_{j=1}^mZ_j$. 

Indeed, the indicators $(\bone_{i\in I})_{i\in[n]}$ are negatively associated, since $I$ is sampled without replacement~\cite{joagdev1983negative}. Consider the auxiliary graph on the left vertices, where $j$ and $j'$ are adjacent whenever
$
N(j)\cap N(j')\neq\emptyset
$.
If a set of left vertices is independent in this auxiliary graph, then the corresponding neighborhoods are disjoint. Hence the associated variables $Z_j$ are monotone non-decreasing functions of disjoint subfamilies of negatively associated variables, and are therefore negatively associated.

Finally, we use the simple expectation bound
$$
\E L_k\le
\E_S\left[\sum_{i\in I}\deg(i)\right]=\sum_{i=1}^n\deg(i)\Pr[i\in I]=\frac{k}{n}\sum_{i=1}^n\deg(i)
=d_Rk.$$
Thus, if $r>d_Rk$, then
$
(r-\E L_k)_+^2\ge (r-d_Rk)^2
$, which completes the proof.
\end{enumerate}
\end{proof}

\subsection{Proof of Corollary~\ref{cor:sunflower-comparison}}\label{app:sunflower-comparison}

\sunflowerComparison*

\begin{proof}
Each component has $2D^2-D+1$ underlying vertices, $D^2+D$ output vertices, and $2D^2$ edges. The common underlying vertex has degree $D$, every other underlying vertex has degree one, and the output degrees are $D$ and one. This proves the stated degree parameters.

In the auxiliary graph, the $D$ nonsingleton outputs give $D$ cliques of size $D$ sharing the common underlying vertex; all singleton inputs are isolated. Give the common vertex one color and use the remaining $D-1$ colors on the private vertices of each clique. Thus this graph is $D$-colorable, and it contains a clique of size $D$, proving that its fractional chromatic number is exactly $D$.

Set $d_0=D$, so that the event $\Lam_k>d_0$ is empty. It suffices to ensure
\[
\ln\frac1{\veps_L}\geq\frac{k}{D}\ln\frac\alpha p,
\]
which is stronger than~\eqref{eq:cond2}. Since $p$ and $\alpha$ are fixed, $\ln(\alpha/p)=\Theta(1)$. For Azuma, the sufficient condition is
\[
\frac{(r-d_Rk)^2}{2kD^2}\geq\frac{k}{D}\ln\frac\alpha p.
\]
For $d_Rk<r$, taking square roots and rearranging gives the choice
\[
k_A=\left\lfloor\frac r{d_R+\sqrt{2D\ln(\alpha/p)}}\right\rfloor
=\Theta(r/\sqrt D).
\]

For Janson, substituting $\chi=D$ in the neighborhood estimate gives the sufficient condition
\[
\frac{2(r-d_Rk)^2}{mD}\geq\frac{k}{D}\ln\frac\alpha p.
\]
Since $d_Rk<r$, this is equivalent to
\[
d_Rk+\sqrt{\frac{mk}{2}\ln\frac\alpha p}\leq r.
\]
It is enough to make each term on the left at most $r/2$. These two requirements are, respectively,
\[
k\leq\frac r{2d_R}
\qquad\text{and}\qquad
k\leq\frac{r^2}{2m\ln(\alpha/p)}.
\]
We therefore choose
\[
k_J=\left\lfloor\min\left\{\frac r{2d_R},\frac{r^2}{2m\ln(\alpha/p)}\right\}\right\rfloor.
\]
If $r=\Theta(m)$, both terms in the minimum are $\Theta(r)$, so $k_J=\Theta(r)$. The bounds above imply~\eqref{eq:cond2} for both choices. Only the event $L_k>r$ contributes a bad term, and hence the proof of Theorem~\ref{thm:general-improvement} gives $2\exp(-(k/D)\D(\alpha\|p))$. These are the claimed tail bounds.

\end{proof}

\subsection{Proof of Corollary~\ref{cor:general improved dL}}\label{app:general improved dL}

\generalLeftImproved*

Let
\(
A:=\frac{M}{d_Rk}
\),
\(
d_0:=\left\lceil
2\sqrt{k\frac{\ln(\alpha/p)}{\ln A}}
\right\rceil 
\).
Since $\ln(\alpha/p)=\Theta(1)$,
\[
\frac{k}{d_0}\D(\alpha\|p)
=
\D(\alpha\|p)\cdot
\Omega\left(
\min\left\{k,\sqrt{k\ln A}\right\}
\right).
\]

First, suppose $\Dl\leq d_0$. Then Theorem~\ref{thm: general baseline} gives
\[
\Pr\left[\sum_{i=1}^nY_i\geq \alpha n\right]
\leq
\exp\left(
-\frac{k}{\Dl}\D(\alpha\|p)
\right)
\leq
\exp\left(
-\frac{k}{d_0}\D(\alpha\|p)
\right),
\]
and the desired bound follows.

Now, suppose $\Dl>d_0$. We verify condition~\eqref{eq:cond1}. By definition of
$d_0$,
\[
d_0^2\ln A
\geq
4k\ln(\alpha/p),
\]
which implies
\[
\frac12 d_0\ln A
\geq
\frac{k\ln(\alpha/p)}{d_0}.
\]
Moreover, since $k\leq r/\Dr$ and $d_R\leq \Dr$, we have $d_Rk\leq r$.
Thus
\(
A=\frac{M}{d_Rk}\geq \frac{M}{r}\to\infty.
\)
Together with the assumption
\(
k\ln A\gg \ln^2(d_Rk),
\)
and $\ln(\alpha/p)=\Theta(1)$, this implies that, for sufficiently large parameters,
\[
d_0\ln A-\ln(d_Rk)
\geq
\frac12d_0\ln A.
\]
Combining the last two inequalities gives
\[
d_0\ln A-\ln(d_Rk)
\geq
\frac{k\ln(\alpha/p)}{d_0}.
\]
Since $\Dl>d_0$, this implies
\[
d_0\ln A-\ln(d_Rk)
\geq
k\left(\frac1{d_0}-\frac1\Dl\right)\ln\frac{\alpha}{p}.
\]
Finally, since $M=Ad_Rk$, we have
\[
(d_0+1)\ln A-\ln M
=
d_0\ln A-\ln(d_Rk).
\]
Therefore,
\[
(d_0+1)\ln\left(\frac{M}{d_Rk}\right)
+\ln((d_0+1)!)
-\ln M
\geq
k\left(\frac1{d_0}-\frac1\Dl\right)\ln\frac{\alpha}{p}.
\]
This is condition~\eqref{eq:cond1}.

Since $k\leq r/\Dr$, the vertices in $I$ are adjacent to at most $k\Dr\leq r$
left vertices. Applying the proof of
Theorem~\ref{thm:general-improvement} with the bad event $L_k>r$ omitted, using only the $\Lam_k$ bad event and
condition~\eqref{eq:cond1}, gives the factor $2$ instead of $3$:
\[
\Pr\left[\sum_{i=1}^nY_i\geq \alpha n\right]
\leq
2\exp\left(
-\frac{k}{d_0}\D(\alpha\|p)
\right).
\]
Substituting the lower bound on $k/d_0$ concludes the proof.

\subsection{Proof of Corollary~\ref{cor:general improved both}}
\label{app:general improved both}

\generalBothImproved*

\begin{proof}
Set
\[
 d_0:=\min\{\Delta_L,\lceil r^{1/3}\rceil\},
 \qquad k=\Theta(d_0r^{1/3}),
\]
where $k$ is an integer and the implied constant is sufficiently small,
depending only on $p$ and $\alpha$. Thus $k=O(r^{2/3})$ and
$k/d_0=\Omega(r^{1/3})$. Since $d_R=O(1)$, we have $d_Rk\leq r/2$
for all sufficiently large $r$. 

We check the two conditions of Theorem~\ref{thm:general-improvement}.
If $d_0=\Delta_L$, condition~\eqref{eq:cond1} is omitted.
Otherwise $d_0=\lceil r^{1/3}\rceil$. Using $d_L\geq1$ and
$m/\Delta_L\geq r$, the left-hand side of~\eqref{eq:cond1} is at least
\begin{align*}
 d_0\ln\frac{r}{d_Rk}-\ln(d_Rk)
 &\geq r^{1/3}\left(\frac13\ln r-O(1)\right)
       -\frac23\ln r-O(1)\\
 &=\Omega(r^{1/3}\ln r).
\end{align*}
This dominates its right-hand side, which is $O(k/d_0)=O(r^{1/3})$.

For condition~\eqref{eq:cond2}, use $\Delta_R^2\leq r$ to obtain
\[
 \frac{(r-d_Rk)^2}{2k\Delta_R^2}
 \geq\frac{r}{8k}
 \geq\frac{k}{d_0}\ln\frac\alpha p.
\]
The last inequality holds by the sufficiently small choice of the
constant in $k=\Theta(d_0r^{1/3})$, since $d_0\leq2r^{1/3}$.
This is stronger than the required condition.
The theorem therefore gives
\[
 \Pr\left[\sum_{i=1}^nY_i\geq\alpha n\right]
 \leq3\exp\left(-\frac{k}{d_0}\D(\alpha\|p)\right)
 =3\exp\bigl(-\Omega(r^{1/3})\bigr),
\]
as $p$ and $\alpha$ are fixed.
\end{proof}

\section{Deferred Proofs of Section~\ref{sec:lin hashing}}
\subsection{Proof of Corollary~\ref{cor:optimal k}}\label{app:max-load-optimization}

\optimalLoadBound*

Applying Theorem~\ref{thm:main-balls} with \(\rho_j=p_{\rm src}\), we obtain
\begin{equation}
\label{eq:Ak-pointwise}
        \Pr[L_y\ge a]
        \le
        A_k:=
        \frac{\binom mk}{\binom ak}
        \frac1{n^k}
        \prod_{j=1}^{k-1}
        \bigl(1+(n-1)q^jp_{\rm src}\bigr).
\end{equation}

We choose $k$ to be
\(
        k:=\left\lf\log_q\frac{a}{mp_{\rm src}}\right\rf
\).
Since \(a=\lceil\alpha m/n\rceil\) and \(\alpha m/n\ge1\), we have
\(\alpha m/n\le a\le2\alpha m/n\), and hence
\[
        \log_q\frac{a}{mp_{\rm src}}
        =
        \log_q\frac1{np_{\rm src}}
        +
        \log_q\alpha+O(\log_q2).
\]
The hypotheses imply \(1\le k\le a\) and \(k^2\le a\).

First, we estimate the binomial term. We get
$$
\ln \frac{\binom mk}{\binom ak} = k\ln\frac{m}{a} + \sum_{i=0}^{k-1}\ln\frac{1-i/m}{1-i/a},
$$
where, since $k^2\leq a$, 
$$
\sum_{i=0}^{k-1}\ln\frac{1-i/m}{1-i/a}\leq -\sum_{i=0}^{k-1}\ln\left(1-\frac{i}{a}\right)\leq 2\sum_{i=1}^{k-1}\frac{i}{a} = \frac{k(k-1)}{a}\leq 1.
$$
Therefore,
\[        \frac{\binom mk}{\binom ak}\frac1{n^k}
        \leq e
        \left(\frac{m}{an}\right)^k
        \leq e\alpha^{-k}.
\]
Using \(k\geq \log_q(1/(np_{\rm src}))+\log_q\alpha-1\), we get
\[
        \alpha^{-k}
        \leq
        \alpha
        (np_{\rm src})^{\log_q\alpha}
        q^{-\log_q^2\alpha}.
\]

It remains to estimate the product term. Let \(
        b:=\log_q\frac1{(n-1)p_{\rm src}}
\). 
Then
\[
        \log_q\prod_{j=1}^{k-1}
        \bigl(1+(n-1)p_{\rm src}q^j\bigr)
        =
        \sum_{j=1}^{k-1}\log_q(1+q^{j-b}).
\]
Since
\(
        \log_q(1+q^x)
        =
        x_+
        +
        \log_q(1+q^{-|x|})\) for $x\in\mathbb R$, where 
\(x_+:=\max\{x,0\}
\), we obtain
\[
        \sum_{j=1}^{k-1}\log_q(1+q^{j-b})
        \le
        \sum_{j=1}^{k-1}(j-b)_+
        +\frac4{\ln q}
        \leq
        \frac12(\log_q\alpha+\log_q2)^2+\frac4{\ln q}.
\]
Indeed, the error is at most
\(\sum_{j\in\mathbb Z}q^{-|j-b|}/\ln q\le4/\ln q\) for
\(q\ge2\). For the last inequality, use
\[
 \sum_{j=1}^{k-1}(j-b)_+
 \le\int_1^k(x-b)_+\,dx\le\tfrac12(k-b)_+^2,
 \qquad
 k-b\le\log_q\frac{a(n-1)}m\le\log_q\alpha+\log_q2.
\]
Thus the product term is at most
\(\sqrt2\,e^4\alpha q^{(\log_q\alpha)^2/2}\).

Combining the above, we obtain the desired bound
\[
        \Pr[L_y\ge \alpha m/n]
        \le
        C\alpha^{2}
        (np_{\rm src})^{\log_q\alpha}
        q^{-\frac12\log_q^2\alpha}.
\]
The maximum-load bound follows by multiplying by \(n\).

\subsection*{Intuition behind the choice of $k$. Comparison with the fully random case}
Since
\begin{equation}\label{eq: Ak ratio}
        \frac{A_{k+1}}{A_k}
        =
        \frac{m-k}{a-k}
        \left(
        \frac1n+\frac{n-1}{n}q^kp_{\rm src}
        \right),
\end{equation}
 \(A_k\) decreases as long as this ratio is less than \(1\).

If \(q^kp_{\rm src}\ll 1/n\), then the product term is negligible and the ratio is
approximately
\(
        \frac{m-k}{n(a-k)}
\). Thus, the optimal choice for $k$ in this case would be
\[
        k
        =
        \frac{na-m}{n-1}
        \approx
        a-\frac mn .
\]
This is exactly the ratio from \cite{siamdm/SchmidtSS95}.

Note that the product term is also $1$ in the case when $X_1,\ldots, X_m$ are fully independent, which, combined with the above choice of $k$, recovers the standard Chernoff-style bound\cite{siamdm/SchmidtSS95}.

Equivalently, it is sufficient that
\[
        \frac1p_{\rm src}\gg nq^{a-m/n}.
\]
In this regime,
\[
        \Pr[L_{\max}\ge a]
        \lesssim
        n\min_{1\le k\le a}
        \frac{\binom mk}{\binom ak}
        \frac1{n^k}.
\]

If \(q^kp_{\rm src}\gg 1/n\), then~\eqref{eq: Ak ratio} leads to
\[
        \frac{m-k}{a-k}q^kp_{\rm src} \approx1,
\]
which reduces to
\[
        k+\log_q\frac{m-k}{a-k}
        =
        \log_q\frac{1}{p_{\rm src}}+O(1).
\]
When \(k=o(a)\), this gives
\[
        k=
        \log_q\frac{a}{mp_{\rm src}}+O(1),
\]
which is the choice used in Corollary~\ref{cor:optimal k}.
\section{Deferred Proofs of Section~\ref{sec:RWs}}
\subsection{The spectral product bound: Proof of Theorem~\ref{thm: spectral bound}}\label{app:spectral-product-bound}

\spectralProductBound*

    Let \(M\) be an ergodic finite-state Markov chain with stationary distribution
\(\pi\).
Let \(f,h:[n]\to[0,1]\) satisfy
\(
\mathbb E_\pi f=\mathbb E_\pi h=\mu
\) and let \(H_f\) and \(H_h\) denote multiplication by \(\sqrt f\) and \(\sqrt h\),
respectively. Then
\[
\|H_fM^gH_h\|_{2,\pi}
\le
\mu+(1-\mu)\lambda(M)^g .
\]
Indeed, let \(u,v\in L_2(\pi)\) with
\(
\|u\|_{2,\pi}=\|v\|_{2,\pi}=1,
\)
and set
\[
a:=H_fu=\bar a\mathbf 1+a_\perp,
\qquad
b:=H_hv=\bar b\mathbf 1+b_\perp,
\]
where
\(
\bar a=\mathbb E_\pi a\), \(
\bar b=\mathbb E_\pi b
\)
and \(a_\perp,b_\perp\) are mean-zero. Since \(M\mathbf 1=\mathbf 1\) and
\(M\) preserves \(\pi\), the cross terms vanish, and therefore
\[
\langle a,M^gb\rangle_\pi
=
\bar a\bar b+\langle a_\perp,M^gb_\perp\rangle_\pi,
\]
and
\[
|\langle a,M^gb\rangle_\pi|
\le
|\bar a||\bar b|
+
\lambda(M)^g\|a_\perp\|_{2,\pi}\|b_\perp\|_{2,\pi}.
\]
By Cauchy–Schwarz inequality,
\[
|\bar a|^2
=
|\langle H_f\mathbf 1,u\rangle_\pi|^2
\le
\|H_f\mathbf 1\|_{2,\pi}^2
=
\mu,
\]
and similarly \(|\bar b|^2\le \mu\). Additionally, since \(0\le f,h\le 1\), we have
\(
\|a\|_{2,\pi}\le 1\), \(
\|b\|_{2,\pi}\le 1.
\)
 Thus
\[
\|a_\perp\|_{2,\pi}^2\le 1-|\bar a|^2,
\qquad
\|b_\perp\|_{2,\pi}^2\le 1-|\bar b|^2.
\]
Therefore,
\[
|\langle a,M^gb\rangle_\pi|
\le
|\bar a||\bar b|
+
\lambda(M)^g\sqrt{1-|\bar a|^2}\sqrt{1-|\bar b|^2}.
\]
Using $\sqrt{1-|\bar a|^2}\sqrt{1-|\bar b|^2}\leq 1- |\bar a||\bar b|$, since it is equivalent to $|\bar a|^2+|\bar b|^2\geq 2|\bar a\bar b|$. Therefore, as $|\bar a\bar b|\leq \mu$, 
we get
\[
|\langle a,M^gb\rangle_\pi|
\le
\mu+(1-\mu)\lambda(M)^g .
\]
This proves
\[
\|H_fM^gH_h\|_{2,\pi}
\le
\mu+(1-\mu)\lambda(M)^g .
\]

Now, fix \(1\le i_1<\cdots<i_k\le t\). For each $i$, let $P_i$ be the diagonal operator with entries $f_i$, and let \(H_i=P_i^{1/2}\). Then by the Markov property,
$$
\E_\varphi\left[\prod_{j=1}^k f_{i_j}(v_{i_j})\right]=
\varphi M^{i_1-1}P_{i_1}M^{g_1}P_{i_2}\cdots M^{g_{k-1}}P_{i_k}\mathbf 1 .
$$

If the walk starts from stationarity, then
\[
\mathbb E_\pi\left[\prod_{j=1}^k f_{i_j}(v_{i_j})\right]
=
\left\langle
H_{i_1}\mathbf 1,
H_{i_1}M^{g_1}H_{i_2}
\cdots
H_{i_{k-1}}M^{g_{k-1}}H_{i_k}\mathbf 1
\right\rangle_\pi .
\]
By Cauchy-Schwarz, and since
\[
\|H_{i_1}\mathbf 1\|_{2,\pi}
=
\|H_{i_k}\mathbf 1\|_{2,\pi}
=
\sqrt{\mu}
\]
we obtain
\[
\mathbb E_\pi\left[\prod_{j=1}^k f_{i_j}(v_{i_j})\right]
\le
\mu\prod_{j=1}^{k-1}
\left(\mu+(1-\mu)\lambda(M)^{g_j}\right).
\]

For a general initial distribution \(\varphi\), let
\(
h:=\frac{d\varphi}{d\pi}
\), so that
\begin{align*}
\mathbb E_\varphi\left[\prod_{j=1}^k f_{i_j}(v_{i_j})\right]&=
\left\langle
h,
M^{i_1-1}P_{i_1}M^{g_1}P_{i_2}\cdots
M^{g_{k-1}}P_{i_k}\mathbf 1
\right\rangle_\pi.\\
&=
\left\langle
(\widetilde M)^{i_1-1}h,
P_{i_1}M^{g_1}P_{i_2}\cdots M^{g_{k-1}}P_{i_k}\mathbf 1
\right\rangle_\pi,
\end{align*}
where \(\widetilde M\) is the \(L_2(\pi)\)-adjoint of \(M\). Since \(\widetilde M\) is the
time-reversal Markov operator, it is a contraction on \(L_2(\pi)\), and hence
\[
\|(\widetilde M)^{i_1-1}h\|_{2,\pi}
\le
\|h\|_{2,\pi}.
\]
Therefore,
\[
\mathbb E_\varphi\prod_{j=1}^k f_{i_j}(v_{i_j})
\le
\|h\|_{2,\pi}\sqrt{\mu}
\prod_{j=1}^{k-1}
\left(\mu+(1-\mu)\lambda(M)^{g_j}\right).
\]
Since
\[
\|h\|_{2,\pi}
=
\left(\sum_x\frac{\varphi(x)^2}{\pi(x)}\right)^{1/2}
=
\|\varphi\|_{\pi^{-1}},
\]
we get
\[
\mathbb E_\varphi\prod_{j=1}^k f_{i_j}(v_{i_j})
\le
\|\varphi\|_{\pi^{-1}}\sqrt{\mu}
\prod_{j=1}^{k-1}
\left(\mu+(1-\mu)\lambda(M)^{g_j}\right).
\]
This completes the proof.
\subsection{Optimization of mixing-time bound in Theorem~\ref{thm: mixing time}}\label{app:mixing-optimization}

\mixingTimeBound*

To find the optimal choice of $k$, we start with $k=1$ and increase it as long as the bound improves. Namely, we increase $k$ while the ratio of the bound for $k+1$ to the bound for $k$ is less than $1$. This expression is equal to
\[
\frac{t_0-k}{\alpha t_0-k}p_{k}\lb\frac{p_{k+1}}{p_k}\rb^{k}
\]
Since $p_{k+1} = p_k+\frac{1-\mu}{t_0}$, we have $(p_{k+1}/p_k)^{k}\lesssim\exp(\frac{(1-\mu)k/t_0}{p_k})$. Hence, we need to solve
$$
\frac{t_0-k}{\alpha t_0-k}
p_{k}
\exp\left(
\frac{(1-\mu)k/t_0}{p_k}
\right)
	\leq1	.
$$
Taking natural logarithms gives
$$
\ln\left[
\left(\mu+(1-\mu)\frac{k}{t_0}\right)
\frac{t_0-k}{\alpha t_0-k}
\right]
+
\frac{(1-\mu)k/t_0}{\mu+(1-\mu)k/t_0}
\leq0.$$
The logarithmic part is
$$
\ln\left[
\left(\mu+(1-\mu)k/t_0\right)
\frac{1-k/t_0}{\alpha-k/t_0}
\right] = \ln\left[
\frac{\mu+(1-2\mu)k/t_0+O((k/t_0)^2)}
{\mu+(\alpha-\mu)-k/t_0}
\right].
$$
Expanding around $\alpha-\mu=0$ and $k/t_0=0$, we get
$$
\ln\left[
\frac{\mu+(1-2\mu)k/t_0+O((k/t_0)^2)}
{\mu+(\alpha-\mu)-k/t_0}
\right]=
-\frac{\alpha-\mu}{\mu}
+
2\frac{1-\mu}{\mu}k/t_0
+
O((\alpha-\mu)^2+(k/t_0)^2).
$$
While the second term is
$$
\frac{(1-\mu)k/t_0}{\mu+(1-\mu)k/t_0}=
\frac{1-\mu}{\mu}k/t_0+O((k/t_0)^2).
$$ Combining, we obtain $$\frac{k}{t_0}\approx\frac{\alpha-\mu}{3(1-\mu)} = \frac{\mu\delta}{3(1-\mu)}.$$
Thus, the optimal choice is $$k=\left\lf\frac{\mu\delta t_0}{3(1-\mu)}\right\rf.$$  This implies a condition $k/t_0=o(1)$.

Now, let us estimate the tail bound. Using $
\ln\binom{n}{pn}=nH(p)+O(\ln n)
$, the natural logarithm of our bound is $$t_0H(k/t_0)-
\alpha t_0H(k/(\alpha t_0))+
(k-1)\ln p_k+
O(\ln t_0).$$
Using the Taylor expansions $H(x)=-x\ln x+x-\frac{x^2}{2}+O(x^3)$ and $\ln(\mu+x)=\ln \mu+\frac{x}{\mu}+O(x^2)$, we get $$H(k/t_0)
=-\frac{k}{t_0}\ln\frac{k}{t_0}
+
\frac{k}{t_0}
\frac{k^2}{2t_0^2}
+
O((k/t_0)^3) \qquad\text{and}\qquad \alpha H(k/(\alpha t_0))
=-\frac{k}{t_0}\ln\frac{k}{\alpha t_0}
+
\frac{k}{t_0}
\frac{k^2}{2\alpha t_0^2}
+
O((k/t_0)^3).$$
Subtracting,
$$
H(k/t_0)-\alpha H(k/(\alpha t_0))=
-\frac{k}{t_0}\ln\alpha
 +
 \frac{k^2}{2t_0^2}
 \left(\frac1\alpha-1\right)
 +
 O((k/t_0)^3).
 $$
Since
$$
\ln\left(\mu+(1-\mu)\frac{k}{t_0}\right)=
\ln\mu
 +
 \frac{1-\mu}{\mu}\frac{k}{t_0}
 +
 O((k/t_0)^2),
 $$
we obtain
$$
\frac{k}{t_0}\ln p_k=
\frac{k}{t_0}\ln\mu
 +
 \frac{1-\mu}{\mu}\frac{k^2}{t_0^2}
 +
 O((k/t_0)^3).
 $$
Combining the above,
$$
 \frac1{t_0}\ln \lb\frac{\binom{t_0}{k}}{\binom{\alpha t_0}{k}}p_k^{k-1}\rb
 \approx
 \frac{k}{t_0}\ln\frac{\mu}{\alpha}
 +
 \frac{k^2}{t_0^2}
 \left[
 \frac12\left(\frac1\alpha-1\right)
 +
 \frac{1-\mu}{\mu}
 \right].
 $$
Since $\alpha=(1+\delta)\mu$,
$$
\ln\frac{\mu}{\alpha}=
-\ln(1+\delta)=
-\delta+O(\delta^2)=
-\frac{\alpha-\mu}{\mu}+O(\delta^2).
 $$
Also, for $\alpha$ close to $\mu$,
$$
\frac12\left(\frac1\alpha-1\right)
+
\frac{1-\mu}{\mu}=
\frac{3(1-\mu)}{2\mu}
 +
 O(\delta),
 $$
 as $1/\alpha-1 = (\mu/\alpha-\mu)/\mu\leq(1-\mu)/\mu$.
Therefore,
$$
 \ln \lb\frac{\binom{t_0}{k}}{\binom{\alpha t_0}{k}}p_k^{k-1}\rb
 \approx
 -\frac{\alpha-\mu}{\mu}k
 +
 \frac{3(1-\mu)}{2\mu}\frac{k^2}{t_0}.
 $$
Using
$$
 \frac{k}{t_0}
 \approx
 \frac{\mu\delta}{3(1-\mu)},
 $$
we get
$$
 \frac1{t_0}\ln \lb\frac{\binom{t_0}{k}}{\binom{\alpha t_0}{k}}p_k^{k-1}\rb
 \approx
 -\delta\cdot \frac{\mu\delta}{3(1-\mu)}
 +
 \frac{3}{2}\frac{1-\mu}{\mu}
 \left(
 \frac{\mu\delta}{3(1-\mu)}
 \right)^2=-\frac{\mu\delta^2}{3(1-\mu)}
+
\frac{\mu\delta^2}{6(1-\mu)}=
-\frac{\mu\delta^2}{6(1-\mu)}.
 $$
That is
$$
 \Pr[X\geq (1+\delta)\mu t]=\Pr[X/T\geq (1+\delta)\mu t_0] 
 \lesssim
 \exp\left(
 -\frac{\mu\delta^2}{6(1-\mu)}t_0
 \right).
 $$
Since $t_0=t/T$,
$$
 \Pr[X\geq (1+\delta)\mu t]
 \leq
 c(\varphi)\cdot
 \exp\left(
 -\frac{\mu\delta^2}{6(1-\mu)}\frac{t}{T}
 +
 o\left(\frac{\mu\delta^2}{1-\mu}\frac{t}{T}\right)
 \right).
 $$

\end{document}